\documentclass[a4paper,11pt]{article}
\usepackage{jheppub} 
\usepackage{lineno}
\usepackage{float}
\usepackage{comment}

\usepackage{graphicx}
\usepackage{feynmp-auto}
\usepackage{tikz-feynman}
\tikzfeynmanset{compat=1.1.0}

\tikzfeynmanset{
  every diagram/.style={
    font=\small,
    baseline=(current bounding box.center)
  },
  every particle/.style={
    font=\scriptsize,
    minimum size=0pt,
    inner sep=1pt
  }
}

\usepackage{amsmath}
\allowdisplaybreaks 

\usepackage{booktabs} 

\usepackage{amssymb,amsmath,amsfonts,nccmath,amsthm}
\usepackage{hyperref}
\usepackage{IEEEtrantools}
\usepackage{float}
\usepackage{graphicx}
\usepackage{xcolor}
\usepackage{tensor}
\usepackage[utf8]{inputenc}
\usepackage{subfig}
\usepackage{bm,bbm,bbold}	
\usepackage{scalerel}[2016/12/29]
\usepackage{enumitem}
\usepackage{multirow}
\usepackage{tabularx}
\usepackage{booktabs}
\usepackage{array, makecell}
\usepackage{mathtools}
\usepackage{algorithm}
\usepackage[noend]{algpseudocode}

\title{\boldmath A Non-Renormalization Theorem for Local Functionals in Ghost-Free Vector Field Theories Coupled to Dynamical Geometry}

\author{Lavinia Heisenberg, }
\author{Shayan Hemmatyar, }
\author{Nadine Nussbaumer }
\affiliation{Institute for Theoretical Physics, Heidelberg University, Philosophenweg 16, 69120 Heidelberg, Germany}

\emailAdd{heisenberg@thphys.uni-heidelberg.de, hemmatyar@thphys.uni-heidelberg.de, nussbaumer\_n@thphys.uni-heidelberg.de}

\abstract{We establish a non-renormalization theorem for a class of ghost-free local functionals describing massive vector field theories coupled to dynamical geometry. Under the assumptions of locality, Lorentz invariance, and validity of the effective field theory expansion below a fixed cutoff, we show that quantum corrections do not generate local operators that renormalize the classical derivative self-interactions responsible for the constraint structure of the theory.
The proof combines an operator-level analysis of the space of allowed local counterterms with a systematic decoupling-limit argument, which isolates the leading contributions to the effective action at each order in the derivative expansion. As a consequence, all radiatively induced local functionals necessarily involve additional derivatives per field and are suppressed by the intrinsic strong-coupling scales of the theory.
In particular, the classical interactions defining ghost-free vector field theories are stable under renormalization, and any additional degrees of freedom arising from quantum corrections appear only above the effective field theory cutoff. This result extends known non-renormalization properties of flat-space vector theories to the case of dynamical geometry and provides a structural explanation for their perturbative stability to all loop orders.}


\begin{document}
\maketitle
\flushbottom

\newtheorem{definition}{Definition}[section]
\newtheorem{example}{Example}[section]
\newtheorem{theorem}{Theorem}[section]
\newtheorem{corollary}{Corollary}[section]
\newtheorem{lemma}{Lemma}[section]
\newtheorem{remark}{Remark}[section]
\newtheorem{proposition}{Proposition}[section]




\section{Introduction}
Local field theories with derivative self-interactions provide a broad class of models in which nontrivial constraint structures can be engineered at the classical level. In particular, suitably chosen derivative couplings may render otherwise higher-derivative theories dynamically consistent by ensuring that the equations of motion propagate only a finite and controlled number of degrees of freedom. From the perspective of effective field theory, such constructions raise a fundamental structural question: whether the constraint structure defining a ghost-free theory is stable under quantum corrections generated by local fluctuations of the fields.
In perturbative quantum field theory, quantum corrections generated by local fluctuations manifest themselves as local and covariant counterterms in the effective action, a fact that can be formulated rigorously within the framework of quantum field theory on curved spacetimes \cite{Brunetti:1999jn,Hollands:2001fb}. Renormalization in curved spacetimes and the algebraic structure of local counterterms can be rigorously formulated using microlocal analysis and locality principles; foundational work in this direction includes the development of microlocal spectrum conditions for quantum field theory on curved backgrounds \cite{Brunetti:1995rf,Brunetti:2001dx} and the construction of local Wick polynomials and time-ordered products \cite{Hollands:2001nf, Hollands:2001fb}. These structural frameworks are deeply connected to the locally covariant approach to quantum field theory, which formalizes how renormalization must respect locality and diffeomorphism covariance in curved geometry \cite{Brunetti:2001dx}. Renormalization can be interpreted as a deformation of the algebra of local observables, see e.g. \cite{Duetsch:2000de} (see also \cite{Ivanov:2024lbs,Ivanov:2024xpr}).

A paradigmatic example of this phenomenon is provided by ghost-free scalar theories with derivative self-interactions, such as Galileon models on flat spacetime \cite{Nicolis:2008in} and their covariant extensions, including Horndeski theories \cite{Horndeski:1974wa}. In these theories, the classical Lagrangian is composed of a finite set of local functionals whose derivative structure guarantees second-order equations of motion and a fixed number of propagating degrees of freedom. Remarkably, despite the absence of a conventional symmetry protecting these interactions, quantum corrections do not renormalize the classical operators. Instead, all radiatively generated local counterterms necessarily involve additional derivatives per field and are suppressed by the intrinsic strong-coupling scales of the theory \cite{Luty:2003vm,Nicolis:2004qq,Burgess:2006bm,Hinterbichler:2010xn,Goon:2016ihr,dePaulaNetto:2012hm,deRham:2012ew,Brouzakis:2013lla,deRham:2013qqa,Heisenberg:2019udf,Heisenberg:2019wjv}. This property is commonly referred to as a non-renormalization theorem and admits a structural interpretation in terms of the organization of the local operator algebra. From a structural perspective, renormalization may be viewed as a deformation of the space of local functionals defining the theory \cite{Costello:2011but}.

Closely related constructions exist for massive vector fields. Ghost-free vector field theories with derivative self-interactions — often referred to as Generalized Proca theories — constitute the most general class of local Lagrangians for a massive vector field that propagate exactly three degrees of freedom \cite{Heisenberg:2014rta,Allys:2015sht,Jimenez:2016isa,Heisenberg:2018vsk}. The absence of ghost-like excitations follows from a nontrivial constraint structure, which renders the temporal component of the vector field nondynamical despite the presence of derivative self-interactions \cite{Heisenberg:2018vsk,Jimenez:2019hpl}. Ghost freedom is ensured by a nontrivial constraint structure in the sense of Dirac \cite{Dirac:1964,Henneaux:1992ig}. When formulated on curved spacetimes, these theories define a natural class of local functionals describing massive vector fields coupled to dynamical geometry. Throughout this work, curved spacetime is treated as a geometric background within the framework of perturbative quantum field theory on curved manifolds \cite{Wald:1994,Hollands:2001fb}.

From an effective field theory perspective, it is essential to determine whether the ghost-free structure of such vector field theories persists under renormalization. Quantum corrections may, in principle, generate additional local operators that modify the derivative structure of the classical action and potentially reintroduce unwanted degrees of freedom. While non-renormalization properties of Generalized Proca theories have been established in flat spacetime \cite{Heisenberg:2020jtr}, the extension of these results to the case of dynamical geometry is nontrivial and remains incomplete. In particular, gravitational fluctuations generically induce new counterterms, and it is not a priori clear whether these preserve the constraint structure responsible for ghost freedom. Partial progress in this direction has been made for restricted subclasses of Proca theories on curved backgrounds, where the divergent part of the one-loop effective action was computed in the absence of derivative self-interactions \cite{Toms:2015fja,Buchbinder:2017zaa,Panda:2021cbf,Ruf:2018gug,Ruf:2018vzq,Garcia-Recio:2019iia} (see also \cite{Salcedo:2022eep} for the non-Abelian case).

In this work, we address this question by establishing a non-renormalization theorem for local functionals defining ghost-free vector field theories coupled to dynamical geometry. Under the assumptions of locality, Lorentz invariance, and validity of the effective field theory expansion below a fixed cutoff, we show that quantum corrections do not generate local operators that renormalize the classical derivative self-interactions. Instead, all radiatively induced local functionals necessarily appear at higher derivative order and are suppressed by the intrinsic strong-coupling scales of the theory. As a consequence, the classical constraint structure is preserved under renormalization, and any additional degrees of freedom induced by quantum effects lie parametrically above the effective field theory cutoff. Our analysis combines an operator-level classification of admissible local counterterms with a systematic decoupling-limit argument that controls the derivative expansion of the effective action. Explicit one-loop computations are employed as consistency checks and to illustrate the general mechanism, while the structural argument extends to all loop orders within the effective field theory framework.

The paper is organized as follows. In Sec.~\ref{sec2}, we review the class of ghost-free vector field theories on curved spacetimes and introduce the relevant local functionals and notation. In Sec.~\ref{InteractionStructure}, we introduce a minimal representative model and derive the cubic and quartic interaction vertices arising from the perturbative expansion of the action. 
In Sec.~\ref{DL_Hierarchy}, we define the decoupling limit as a scaling limit of the effective action and explain its role in controlling the derivative expansion and the structure of admissible counterterms.
Sec.~\ref{quantumcorrections} analyzes representative one-loop contributions to the two-point, three-point and, four-point functions of the vector and metric fields, combining power-counting arguments with explicit evaluation of ultraviolet divergences and their interpretation within the decoupling-limit framework.
In Sec.~\ref{NR-Theorem}, we formulate and prove the non-renormalization theorem for local functionals defining ghost-free vector field theories coupled to dynamical geometry. Sec.~\ref{robustness} extends the structural argument to the full class of such theories and to arbitrary loop order within the effective field theory regime. Sec.~\ref{conclusion} contains our conclusions.
Throughout this work we employ the mostly-minus metric signature $(+,-,-,-)$, follow Wald’s conventions for curvature tensors \cite{Wald:1984rg}, and work in units $c=G=1$.




\section{Ghost-Free Vector Field Theories on Curved Spacetimes}
\label{sec2}
We consider classical field theories defined on a four-dimensional Lorentzian manifold $(\mathcal{M}, g_{\mu\nu})$, whose dynamical fields consist of a massive vector field $A_\mu$ and the spacetime metric $g_{\mu\nu}$. The dynamics is specified by local, diffeomorphism-covariant functionals of these fields and their derivatives, organized as a derivative expansion and interpreted within an effective field theory framework  \cite{Heisenberg:2014rta}. Throughout this work, locality is understood in the standard sense that the action functional depends on the fields and a finite number of their derivatives at each spacetime point, while higher-derivative contributions are suppressed by a set of intrinsic strong-coupling scales that define the regime of validity of the effective description.

A central structural requirement imposed on the class of theories studied here is that they propagate a fixed and finite number of physical degrees of freedom. In particular, despite the presence of derivative self-interactions, the phase-space dynamics is required to propagate exactly three degrees of freedom associated with the massive vector field. This property is ensured by a nontrivial constraint structure, in the sense of Dirac’s theory of constrained systems, which renders the temporal component $A_0$ nondynamical and removes the would-be ghost additional degree of freedom from the spectrum. The absence of ghost-like excitations is therefore understood as a consequence of the constraint algebra, rather than as a result of fine-tuning or symmetry protection \cite{Heisenberg:2014rta,Jimenez:2016isa,Heisenberg:2018vsk}.

\subsection{Structural framework and locality assumptions}
\begin{definition}[Ghost-free vector field theories on curved spacetimes]
\label{def:ghostfreevector}
Let $(\mathcal{M}, g_{\mu\nu})$ be a four-dimensional Lorentzian manifold. A \emph{ghost-free vector field theory on curved spacetime} is defined by an action functional
\[
S[A,g] = \int_{\mathcal{M}} \mathcal{L}(A_\mu, \nabla_\nu A_\mu, g_{\mu\nu}, R_{\mu\nu\rho\sigma}, \dots),
\]
satisfying the following properties:
\begin{enumerate}
\item[(i)] \emph{Locality and covariance:} $\mathcal{L}$ is a local, diffeomorphism-covariant scalar density constructed from the fields and a finite number of their derivatives.
\item[(ii)] \emph{Derivative expansion:} Higher-derivative terms are organized in a derivative expansion and suppressed by a finite set of strong-coupling scales, defining an effective field theory.
\item[(iii)] \emph{Constraint structure:} The associated Hamiltonian system propagates exactly three physical degrees of freedom associated with the massive vector field.
\end{enumerate}
\end{definition}\medskip
\begin{proposition}
\label{prop:proca_realization}
The Generalized Proca theories provide a concrete realization of Definition~\ref{def:ghostfreevector}.
In particular, their interaction terms constitute the \textbf{most general} local functionals of a massive vector field and the metric that satisfy locality, diffeomorphism covariance, and propagate exactly three vector degrees of freedom at the classical level and have equations of motion of second differential order\cite{Heisenberg:2014rta}.
\end{proposition}
\noindent
A detailed construction of these interactions can be found in Refs.~\cite{Heisenberg:2014rta,Jimenez:2016isa,Heisenberg:2018vsk}. Throughout this work, we restrict attention to local, perturbative aspects of the theories defined above. In particular, global issues, boundary conditions, and nonperturbative effects are not addressed. All statements are understood within the regime of validity of the derivative expansion and refer to local properties of the action functional and its perturbative quantization.

The corresponding action of the Generalized Proca interactions can be written as a finite sum \cite{Heisenberg:2014rta, Heisenberg:2018vsk}
\begin{equation}
\label{gpaction}
S=\int d^4 x \sqrt{-g}\mathcal{L}, \quad \mathcal{L}=\sum_{i=2}^6 \mathcal{L}_i,
\end{equation}
where the individual Lagrangian densities are given by
\begin{align}
\label{gpLagr}
    &\mathcal{L}_2  =G_2(X,Y,F),\notag \\[5pt]
    &\mathcal{L}_3 =G_3(X) \nabla_\mu A^\mu,\notag \\[4pt]
    &\mathcal{L}_4  =G_4(X) R+\,G_{4, X}\left[\left(\nabla_\mu A^\mu\right)^2-\nabla_\rho A_\sigma \nabla^\sigma A^\rho\right]\notag, \\[4pt]
    &\mathcal{L}_5  =G_5(X) G_{\mu \nu} \nabla^\mu A^\nu-\frac{1}{6} G_{5, X}\left[(\nabla \cdot A)^3+2 \nabla_\rho A_\sigma \nabla^\gamma A^\rho \nabla^\sigma A_\gamma-3(\nabla \cdot A) \nabla_\rho A_\sigma \nabla^\sigma A^\rho\right]\notag\\
    &\qquad\ -g_5(X) \tilde{F}^{\alpha \mu} \tilde{F}^\beta{ }_\mu \nabla_\alpha A_\beta, \notag\\[4pt]
    &\mathcal{L}_6 =G_6(X) \mathcal{L}^{\mu \nu \alpha \beta} \nabla_\mu A_\nu \nabla_\alpha A_\beta+\frac{1}{2}G_{6, X} \tilde{F}^{\alpha \beta} \tilde{F}^{\mu \nu} \nabla_\alpha A_\mu \nabla_\beta A_\nu.
\end{align}
The \(G_i\) are coupling functions that depend on the following independent scalar quantities
\begin{equation}
\label{XYF}
    X\equiv -A_\mu A^\mu/2,\quad Y\equiv A^\mu A^\nu F^\alpha_\mu F_{\nu\alpha},\quad F\equiv -\frac{1}{4}F_{\mu\nu}F^{\mu\nu},
\end{equation}
where $X$ describes the vector field norm, $F$ the kinetic Proca term, and $Y$ is constructed such that it does not give any dynamics to the zeroth component of the vector. Derivatives of the couplings $G_i$ with respect to $X$ are denoted by  $G_{i,X}$. These functions couple the Proca field non-minimally to gravity through the Ricci scalar $R$, the Einstein tensor $G_{\mu\nu}$ and the double dual Riemann tensor \( \mathcal{L}^{\mu\nu\alpha\beta}\equiv\frac{1}{4}\epsilon^{\mu\nu\rho\sigma} \epsilon^{\alpha\beta\gamma\delta} R_{\rho\sigma\gamma\delta} \). The Lagrangian $\mathcal{L}_2$ contains both the kinetic term and a vector potential, which particularly includes the mass term, and allows for all possible contractions among the vector field $A_\mu$ and the associated covariant field strength tensor $F_{\mu\nu}=\nabla_\mu A_\nu-\nabla_\nu A_\mu$, whose dual reads $\tilde{F}_{\mu\nu}\equiv\frac{1}{2}\epsilon^{\mu\nu\rho\sigma}F_{\rho\sigma}$, where we have ignored parity-violating terms that come from contractions with the dual field strength. The derivative self-interactions in the Lagrangians $\mathcal{L}_{i>2}$ are organized such that the resulting Proca equations of motion manifestly remain second-order, avoiding ghost instabilities while maintaining three dynamical vector degrees of freedom. Particularly, by focusing on the longitudinal scalar mode of the vector field $A_\mu\rightarrow\nabla_\mu \phi$, one recovers the scalar-tensor Horndeski theories \cite{Horndeski:1974wa}. At the level of the Proca interactions, both the Lagrangian $\mathcal{L}_6$ and the contribution proportional to $g_5(X)$ in $\mathcal{L}_5$ are genuinely new and pure vector interactions that vanish in the Horndeski limit \cite{Heisenberg:2018vsk}.
\begin{remark}[Radiative stability from structural properties]
The analysis of radiative stability depends only on the structural properties stated in Definition~\ref{def:ghostfreevector}, and not on the detailed form of the chosen interaction terms. They serve as a representative example rather than as an essential ingredient of the proof.
\end{remark}

\subsection{Background Equations of Motion}
In order to define the perturbative expansion used in the subsequent analysis, we first characterize the stationary points of the action functional introduced above. Varying the action with respect to the vector field 
$A_\mu$ and the metric $g_{\mu\nu}$ yields a coupled system of Euler–Lagrange equations, whose explicit form depends on the chosen representative realization of Definition~\ref{def:ghostfreevector}. For the purposes of this work, we restrict attention to backgrounds that solve these equations and admit a well-defined perturbations analysis of the theory.

\begin{definition}[Perturbative expansion about a background configuration]
\label{def:perturbativeexpansion}
Let $(\bar g_{\mu\nu},\bar A_\mu)$ be a background configuration solving the
Euler--Lagrange equations of the action \eqref{gpaction}. The perturbative fields
$(h_{\mu\nu},\delta A_\mu)$ are defined by
\begin{align}
g_{\mu\nu} &= \eta_{\mu\nu} + h_{\mu\nu},\,\,\,\,\,\,\left| h_{\mu\nu} \right| \ll 1, \\
A^\mu &= \bar{A}^\mu + \delta A^\mu,\,\,\,\,\,\,\left|\delta A^\mu \right| \ll 1.
\end{align}
The vector field \(A^\mu\) is decomposed into a background component \(\bar{A}^\mu\) and small perturbations \(\delta A^\mu\). Likewise, \(h_{\mu\nu}\) represents the graviton field emerging as metric perturbation around flat spacetime, where \(\eta_{\mu\nu}\) denotes the Minkowski metric.
\end{definition}

 \begin{proposition}[Euler--Lagrange equations for the background fields]
\label{prop:EOM}
Let $S[A,g]$ be an action functional belonging to the class defined in
Definition~\ref{def:ghostfreevector}.
Then stationary points $(\bar A_\mu,\bar g_{\mu\nu})$ satisfy the coupled
Euler--Lagrange equations obtained by variation with respect to $A_\mu$ and
$g_{\mu\nu}$,
\begin{equation}
\frac{\delta S}{\delta A_\mu}\Big|_{(\bar A,\bar g)} = 0, \qquad
\frac{\delta S}{\delta g_{\mu\nu}}\Big|_{(\bar A,\bar g)} = 0 .
\end{equation}
For the representative realization given by the Generalized Proca interactions \eqref{gpaction},
these equations take the explicit form
\begin{equation}
\frac{1}{2} \bar{G}_2 \eta_{\mu \nu} -\frac{1}{2} \bar{A}_{\mu} \bar{A}_{\nu} \bar{G}_{2,X}=0,
\label{metric_BGEOM}
\end{equation}
for the metric field equations, and
\begin{equation}
\bar{A}^\mu \bar{G}_{2,X}=0.
\label{proca_BGEOM}
\end{equation}
for the vector field equations of motion. All the other functionals trivialise for the flat metric configuration.
\end{proposition}
Here, quantities evaluated on the background are denoted with an overbar.
\begin{lemma}[Flat background solutions]
\label{lem:flatbackground}
The Minkowski metric $\bar g_{\mu\nu} = \eta_{\mu\nu}$ solves the Euler--Lagrange
equations \eqref{metric_BGEOM} and \eqref{proca_BGEOM} associated with the action
\eqref{gpaction} for the following branches of background configurations:
\begin{align}
&(i)\quad \bar A^\mu = 0, \qquad \bar G_2 = 0, \nonumber\\
&(ii)\quad \bar G_2 = 0, \qquad \bar G_{2,X} = 0,
\label{bgconfig}
\end{align}
where $\bar G_2$ and $\bar G_{2,X}$ denote the background values of the function
$G_2$ and its derivative with respect to $X$.
\end{lemma}
While the second branch $(ii)$ allows for a non-vanishing Proca background $\bar{A}^\mu\neq 0$, this comes at the expense of local Lorentz invariance violation. For example, such a configuration could arise from a ``Mexican-hat-type'' potential of the form $G_2(X)=(X-\frac{1}{2}\bar{A}^2)^2$, which realizes the additional constraint on $\bar{G}_{2,X}$ and spontaneously breaks the local Lorentz symmetry. However, a non-zero background vector field would carry over to the perturbative expansion and dress the resulting propagators and interaction vertices, thus greatly increasing the theory's complexity with exceeding computational costs. For this reason, we focus on the first solution branch $(i)$ and, from now on, assume a vanishing background configuration for the Proca field to maintain local Lorentz symmetry.

\begin{remark}[stationary point of the action]
The choice of a Minkowski background is made for technical simplicity and suffices to capture the local ultraviolet structure of the theory relevant for the analysis of counterterms.
\end{remark}

\subsection{Perturbative Expansion: Quadratic Order}
In this section, we consider the second variation of the action functional associated with the Generalized Proca interactions \eqref{gpLagr} around the background configuration \eqref{bgconfig}.
In particular, we perform a perturbative expansion in the fluctuations \( \delta A_\mu \) and \( h_{\mu\nu} \) about the Minkowski background, retaining terms up to quadratic order that capture the dynamics of the Proca field and the graviton. 
\begin{proposition}[Quadratic action and second variation]
\label{prop:quadraticaction}
The quadratic action governing the dynamics of perturbations
$(h_{\mu\nu},\delta A_\mu)$ is given by the second variation of the action
functional \eqref{gpaction} evaluated at the background configuration,
\begin{equation}
S^{(2)}[h,\delta A]
= \frac12 \,\delta^2 S\big|_{(\bar g,\bar A)} .
\end{equation}
For Generalized Proca interactions, $S^{(2)}$ is a local functional involving at
most two derivatives acting on each perturbative field.
\end{proposition}
After suitable integration by parts, the resulting quadratic contributions read
\begin{equation}
\begin{aligned}
\mathcal{L}_2^{(2)} &= 
-\frac{1}{2}\bar{G}_{2,F} \, \partial^\nu\delta A^\mu \,  \partial_\nu \delta A_\mu
+ \frac{1}{2}\bar{G}_{2,F} \, \partial^\nu\delta A^\mu \, \partial_\mu \delta A_\nu
- \frac{1}{2} \bar{G}_{2,X}\, \delta A_\mu \delta A^\mu \,  \,, \\[6pt]
\mathcal{L}_3^{(2)} &= 0 \,, \\[6pt]
\mathcal{L}_4^{(2)} &= 
- \frac{1}{2} \bar{G}_4 \, h^{\mu\nu} \, \partial^\rho \partial_\nu h_{\mu\rho}
+ \frac{1}{2} \bar{G}_4 \, h^\mu{}_\mu \, \partial^\rho \partial^\nu h_{\nu\rho}
+ \frac{1}{4} \bar{G}_4 \, h^{\mu\nu} \, \partial^2 h_{\mu\nu}
- \frac{1}{4} \bar{G}_4 \, h^\mu{}_\mu \, \partial^2 h^\nu{}_\nu \,, \\[6pt]
\mathcal{L}_5^{(2)} &= 0 \,, \\[6pt]
\mathcal{L}_6^{(2)} &= 0 \,.
\end{aligned}
\end{equation}
where $\partial^2\equiv\eta^{\mu\nu}\partial_\mu\partial_\nu$ denotes the flat-space d'Alembertian. Clearly, only the Lagrangians $\mathcal{L}_2$ and $\mathcal{L}_4$ contribute, as they contain the kinetic terms for the Proca field and the graviton, whereas the other Lagrangians purely describe interactions and hence vanish at quadratic order in perturbations. After canonical normalization of the fields  
\begin{equation}
\label{can}
    \delta A_\mu\equiv\frac{\delta A_\mu^{\rm{can}}}{\sqrt{\bar{G}_{2,F}}},\,\,\,\,\,h_{\mu\nu}\equiv \sqrt{\frac{2}{\bar{G}_4}}h_{\mu\nu}^{\rm{can}},
\end{equation}
and introducing the field strength tensor $F^{\delta A}_{\mu\nu}\equiv\partial_\mu\delta A_\nu-\partial_\nu\delta A_\nu$ on flat space, the quadratic Lagrangian becomes

\begin{equation}
\mathcal{L}^{(2)}_{\rm{can}} = -\frac{1}{4} F_{\mu\nu}^{\delta A}F^{\mu\nu}_{\delta A}
- \frac{1}{2} \, \frac{\bar{G}_{2,X}}{\bar{G}_{2,F}} \, \delta A^\mu \delta A_\mu + \frac{1}{2}h_{\mu\nu}\mathcal{O}^{\mu\nu\alpha\beta}h_{\alpha\beta},
\label{Lcan}
\end{equation}
where we omit the superscript $(\cdot)^{can}$ from now on, assuming that all quantities are canonically normalized. Furthermore, we have introduced the graviton kinetic operator \cite{Hinterbichler:2011tt}
\begin{equation}
    \mathcal{O}^{\mu\nu}_{\alpha\beta}\equiv (\eta_\alpha^{(\mu}\eta_\beta^{\nu)}-\eta^{\mu\nu}\eta_{\alpha\beta})\partial^2-2\partial^{(\mu}\partial_{(\alpha}\eta_{\beta)}^{\nu)}+\partial^\mu\partial^\nu\eta_{\alpha\beta}+\partial_\alpha\partial_\beta\eta^{\mu\nu}.
    \label{kinop}
\end{equation}
The graviton sector in \eqref{Lcan} is invariant under linearized diffeomorphisms of the form
\begin{equation}
h_{\mu\nu} \rightarrow h_{\mu\nu} + \partial_\mu \xi_\nu + \partial_\nu \xi_\mu \,,
\end{equation}
where \( \xi_\mu \) is an arbitrary vector field parametrizing infinitesimal coordinate transformations. This gauge symmetry introduces redundancies in the description of the graviton degrees of freedom, rendering the kinetic operator \eqref{kinop} non-invertible. To properly define the graviton propagator, one must add a gauge-fixing term to the kinetic Lagrangian in order to eliminate the gauge degrees of freedom. Employing the harmonic gauge implemented in a Fadeev-Popov procedure \cite{Hinterbichler:2011tt}
\begin{equation}
\partial^\mu h_{\mu\nu} - \frac{1}{2} \partial_\nu h = 0,
\end{equation}
requires adding the following gauge-fixing Lagrangian
\begin{equation}
\mathcal{L}_{\text{GF}} = -\left( \partial^\mu h_{\mu\nu} - \frac{1}{2} \partial_\nu h \right)^2 \,.
\end{equation}
This term removes the degeneracy of the graviton's kinetic operator in \eqref{Lcan} and allows for a consistent definition of the graviton propagator.

Including the gauge-fixing contribution and performing integrations by parts, the final quadratic Lagrangian takes the following form
\begin{equation}
\begin{aligned}
\label{Lkin}
\mathcal{L}_{\rm{kin}}^{(2)} =\ & \mathcal{L}_{\rm{can}}^{(2)}+\mathcal{L}_{\rm{GF}},\\[6pt]
=\ &-\frac{1}{4} F_{\mu\nu}^{\delta A}F^{\mu\nu}_{\delta A}
+ \frac{1}{2} \, m_{\rm{eff}}^2 \, \delta A^\mu \delta A_\mu + \frac{1}{2} \, h^{\mu\nu} \, \partial^\rho \partial_\rho h_{\mu\nu}
- \frac{1}{4} \, h \, \partial^\rho \partial_\rho h \,,
\end{aligned}
\end{equation}
where we have denoted the effective Proca mass by $m_{\rm{eff}}^2\equiv -\bar{G}_{2,X}/\bar{G}_{2,F}$, whose value depends on the precise theory at hand\footnote{In theories with standard kinetic and mass term, this effective mass just reduces to the bare Proca mass, as we expand around a vanishing Proca background with $\bar{X}=0$ and $\bar{F}=0$, such that $\bar{G}_{2,F}=1$ and $\bar{G}_{2,X}=-m^2$.}. Here, the first two contributions correspond to the kinetic and mass contribution of the Proca field. The last two terms describe the graviton kinetic term after including the gauge-fixing contribution. The resulting quadratic action makes manifest the constraint structure responsible for the fixed degree-of-freedom count required by Definition~\ref{def:ghostfreevector} and provides the starting point for the perturbative quantization and the analysis of radiative corrections carried out in later sections. 
\begin{corollary}[Perturbative degrees of freedom]
\label{cor:DOF}
The quadratic action $S^{(2)}$ propagates two tensorial degrees of freedom
associated with the metric perturbations and three vectorial degrees of freedom
associated with the Proca field, with no additional ghost-like excitations.
\end{corollary}

From the above quadratic Lagrangian, the momentum-space propagators for the canonically normalized fields are readily found
\begin{equation}
\begin{aligned}
 &\hspace{10pt}\feynmandiagram[
   inline=(a.base),
   horizontal=a to b
 ] {
  a [particle=\(\mu\)]
    -- [fermion, momentum=\(p\)]
  b [particle=\(\nu\), right=18mm of a],
 };
 \hspace{32pt}
 D_{\mu\nu}(p) =
 \frac{i}{p^2 - m_{\rm eff}^2}
 \left( -\eta_{\mu\nu}
 + \frac{p_\mu p_\nu}{m_{\rm eff}^2} \right),\\[10pt]
 &\feynmandiagram[
   inline=(a.base),
   horizontal=a to b
 ] {
  a [particle=\(\mu\nu\)]
    -- [gluon, momentum=\(p\)]
  b [particle=\(\rho\sigma\), right=20mm of a],
 };
 \qquad
 D_{\mu\nu\rho\sigma}(p) =
 \frac{-i}{p^2}
 \frac{
   \eta_{\mu\rho} \eta_{\nu\sigma}
 + \eta_{\mu\sigma} \eta_{\nu\rho}
 - \eta_{\mu\nu} \eta_{\rho\sigma}
 }{2}.
\end{aligned}
\end{equation}
with implicit Feynman prescription. Notice that at high energies $E\gg m_{\rm{eff}}$, the vector propagator goes like $\sim 1/m_{\rm{eff}}^2$ rather than $\sim 1/p^2$, which invalidates standard power counting arguments. The good high-energy behaviour can be restored through a gauge-fixing procedure, which we relegate to Sec. \ref{DL_Hierarchy}.

A key conceptual ingredient in our analysis is the decoupling limit, which we employ not merely as a computational simplification, but as a structural tool for controlling quantum corrections. The decoupling limit defines a precise scaling regime in which the strong-coupling scales of the theory are held fixed while all other mass scales are parametrically separated. In this limit, the longitudinal mode of the Proca field and its interactions with the graviton isolate the maximally sensitive sector of the theory, in the sense that any local operator capable of destabilizing the effective field theory below its cutoff must already appear in this limit. Consequently, the classification of radiatively generated operators in the decoupling limit provides a complete characterization of all admissible quantum counterterms within the effective field theory. The absence of lower-derivative or unsuppressed operators in this limit therefore implies a non-renormalization of the classical interactions and establishes radiative stability to all orders in perturbation theory.




\section{Interaction Structure}
\label{InteractionStructure}
To facilitate the perturbative expansion at higher orders, we now focus on a specific subset of the Generalized Proca theories by adopting the following minimal model describing a massive vector field minimally coupled to gravity

\begin{equation}
\begin{aligned}
    &G_2(X, F) = F - m^2X , \\
    &G_3(X) = -\frac{m^2}{\Lambda_2^2} c_3\, X, \\
    &G_4(X) = \frac{M_\text{pl}^2}{2},\\
    &G_5(X) = g_5(X)=G_6(X)=0,
    \label{minmodel}
\end{aligned}
\end{equation}
where \( c_3 \) is a dimensionless constant of order unity and the prefactor of $-1$ in $G_3$ was introduced for mere convenience. The classical scales of the theory are the vector mass $m$, the Proca effective field theory cutoff $\Lambda_2$ and the reduced Planck mass $M_\text{pl}$ that are parametrically far separated $m\ll \Lambda_2\ll M_{\rm{pl}}$, an important assumption for the Generalized Proca effective field theory to be well-defined. In that sense, the dimensionless factor $m/ \Lambda_2$ can be regarded as a small coupling constant controlling the size of the vector interactions. The coupling functions are dimensionful quantities and are chosen such that the resulting Lagrangians have the right dimension. Notice that the gravitational sector is purely described by the Einstein-Hilbert action, which admits a perturbative expansion of the form

\begin{equation}
    \mathcal{L}_4\sim \sum_{n=0}^\infty(\partial h)^2\left(\frac{h}{M_{\rm pl}}\right)^n,
\end{equation}
meaning that the Proca field is coupled minimally to gravity. Therefore, non-trivial vector-tensor interactions will appear only through perturbative expansions of the covariant derivatives in $F$ and the metric determinant in the action \eqref{gpaction}. Schematically, these terms can be expanded as
\begin{equation}
    \sqrt{-g}\sim 1+\frac{h}{M_{\rm pl}},\qquad F\sim (\nabla A)^2\sim(\partial\,\delta A)^2+\frac{1}{M_{\rm pl}}(\partial\,\delta A)(\partial h)\delta A+\frac{1}{M_{\rm pl}^2}(\partial h)^2\delta A^2,
\end{equation}
where we expanded $F$ up to quartic order, as higher-order terms would simply introduce additional factors of $\sim \frac{h}{M_{\rm pl}}$. Although the graviton is heavily suppressed by the Planck mass, the non-linear way through which it enters in these expansions can generate non-trivial mixings of the different scales and thus spoil the critical effective field theory hierarchy. Nevertheless, in Sec. \ref{DL_Hierarchy} we will establish  the natural hierarchy $m \ll \Lambda_3 \ll \Lambda_4 \ll \Lambda_2 \ll M_{\mathrm{pl}}$. Focusing on expanding the Generalized Proca Lagrangians in model \eqref{minmodel} to cubic order in perturbations, yields

\begin{equation}
\begin{aligned}
\label{Lcubic}
    \mathcal{L}^{(3)}_{2} &= - \frac{1}{4M_{\rm pl}}\, h^\mu{}_\mu\, F^{\delta A}_{\beta\alpha}\, F_{\delta A}^{\beta\alpha}+
    \frac{m^2}{2M_{\rm pl}}\, h^\mu{}_\mu\, \delta A_\alpha\, \delta A^\alpha
    + \frac{m^2}{M_{\rm pl}}\, h_{\alpha\beta}\, \delta A^{\alpha}\, \delta A^{\beta}\\
    &\quad - \frac{1}{M_{\rm pl}}\, h_{\alpha\gamma}\, (\partial_{\beta} \delta A^{\gamma})\, (\partial^{\beta} \delta A^{\alpha})
    + \frac{1}{M_{\rm pl}}\, h_{\beta\gamma}\, (\partial^{\beta} \delta A^{\alpha})\, (\partial^{\gamma} \delta A_{\alpha})\\
    &\quad + \frac{2}{M_{\rm pl}}\, \delta A^{\alpha}\, (\partial_{\beta} h_{\alpha\gamma})\, (\partial^{\gamma} \delta A^{\beta})
    - \frac{2}{M_{\rm pl}}\, \delta A^{\alpha}\,( \partial_{\gamma} h_{\alpha\beta})\, (\partial^{\gamma} \delta A^{\beta}),  \\[5pt]
    \mathcal{L}^{(3)}_{3} &= 
    \frac{m^2}{2\Lambda_2^2}\,c_3\, \delta A_\alpha\, \delta A^\alpha\,(\partial_\beta \delta A^\beta), 
\end{aligned}
\end{equation}
In $\mathcal{L}_3^{(3)}$, we recognize the pure cubic Proca interaction from flat space, whereas expanding $\mathcal{L}_2$ to third order brought about genuinely new interactions mixing the vector and tensor sector. Evidently, the kinetic and mass term of the Proca field couple to the graviton through perturbing the metric determinant, but the field strength tensor $F_{\mu\nu}$ also gets broken apart as we perturb the inverse metric used to raise indices in $F$ (see Eq. \eqref{XYF}), resulting in terms $\sim h(\partial\delta A)^2$, and through expanding the covariant derivatives, yielding terms of the form $\sim (\partial h) \delta A$. 

Calculating the Generalized Proca Lagrangians at fourth order in perturbations yields, after suitable integration by parts, the following quartic interactions

\begin{align} 
\mathcal{L}^{(4)}_{2} &\sim 
    \frac{m^2}{M_{\rm{pl}}^2}\,h^2\delta A^2+\frac{1}{M_{\rm{pl}}^2}\,h^2 F_{\delta A}^2+ \frac{1}{M_{\rm{pl}}^2}\,h^2(\partial \delta A)^2+ \frac{1}{M_{\rm{pl}}^2}\,h\,(\partial \delta A) (\partial h) \delta A+ \frac{1}{M_{\rm{pl}}^2}\,(\partial h)^2\delta A^2,\notag\\[10pt]
\label{Lquartic}
\mathcal{L}^{(4)}_{3} &\sim 
    \frac{m^2}{2M_{\rm pl}\Lambda_2^2}\,h\, \delta A^2\,(\partial \delta A),
\end{align}
where we have suppressed the Lorentz indices for better readability and give the explicit Lagrangians in Appendix \ref{app1}.

\begin{remark}[Generality of the interaction terms]
\label{rem:generality}
The specific form of the interaction terms \eqref{minmodel} introduced in this section should be understood as representative of the general class of ghost-free Generalized Proca effective field theories, rather than as the result of additional fine tuning. Our conclusions do not rely on special choices of coefficients beyond those required to ensure the absence of Ostrogradsky instabilities at the classical level. In particular, the non-renormalization and radiative stability results derived in the subsequent sections depend only on the derivative structure of the interactions, the hierarchy of strong-coupling scales, and the existence of a well-defined decoupling limit.
\end{remark}

Any ghost-free Generalized Proca theory related to the present formulation by field redefinitions or by the addition of operators suppressed by the effective field theory cutoff exhibits the same pattern of quantum corrections and satisfies the same non-renormalization properties.




\section{Hierarchy of Scales and Decoupling Limit}
\label{DL_Hierarchy}
On flat space, Generalized Proca theories are closely related to their scalar sibling, the Galileon theories. This relationship becomes evident at high energies far beyond the Proca mass, where the dynamics is governed by the longitudinal mode which decouples from the transverse polarizations, such that large parts of the interactions reduce to their corresponding Galileon counterparts. To make this behaviour explicit, it is convenient to perform the Stückelberg trick \cite{Stueckelberg:1900zz, Hinterbichler:2011tt, Heisenberg:2018vsk, Heisenberg:2020jtr}
\begin{equation}
\label{stb}
    \delta A_\mu \rightarrow \delta A_\mu + \frac{1}{m}\partial_\mu \phi,
\end{equation}
where the additional (canonically normalized) scalar field $\phi$ is introduced to restore a formal gauge symmetry that was broken explicitly by the Proca mass term. The scalar plays the role of the associated Goldstone boson which non-linearly realizes the gauge symmetry, while the field strength tensors $F_{\mu\nu}^{\delta A},\,\tilde{F}_{\mu\nu}^{\delta A}$ remain invariant under the replacement \eqref{stb}, as it is modelled after the broken gauge symmetry. The resulting theory is now invariant under the following joint gauge redundancy
\begin{equation}
    \label{gaugesymm}
    \phi\rightarrow \phi+m\,\alpha,\quad \delta A_\mu\rightarrow \delta A_\mu-\partial_\mu\alpha.
\end{equation}
By choosing the unitary gauge $\alpha=-\frac{\phi}{m}$, we obtain $\phi=0$ such that  the original minimal model \eqref{Lkin}, \eqref{Lcubic}-\eqref{Lquartic} is restored,  demonstrating that the two theories are identical. Employing instead the Lorentz-like gauge $\partial_\mu\delta A^\mu+m\,\phi=0$ in a Fadeev-Popov procedure diagonalizes the Stückelberg-ed kinetic Lagrangian $\mathcal{L}_{2}^{(2)}$ and hence results in the following propagators for $\delta A^\mu$ and $\phi$ \cite{Hinterbichler:2011tt}
\begin{equation}
    D_{\mu\nu}(p)=\frac{-i\,\eta_{\mu\nu}}{p^2+m^2},\qquad D(p)=\frac{-i}{p^2+m^2},
\end{equation}
which are well-behaved in the high-energy limit $p^2\gg m^2$. 
Applying the Stückelberg trick to the Generalized Proca interactions \eqref{Lcubic}-\eqref{Lquartic} gives a number of non-renormalizable terms of the following schematic form
\begin{equation}
\begin{aligned}
    \label{LpreDL}
    &\mathcal{L}_2^{(3)}\supset \frac{1}{M_{\rm{pl}}m} \, h\,(\partial^2\phi)F_{\delta A}+\frac{1}{M_{\rm{pl}}m}(\partial h)(\partial\phi)F_{\delta A},\\[5pt]
    &\mathcal{L}_3^{(3)}\supset \frac{c_3}{\Lambda_2^2m}(\partial \phi)^2(\partial^2\phi),\\[5pt]
    &\mathcal{L}_2^{(4)}\supset \frac{1}{M_{\rm{pl}}^2m^2}\,h^2\,(\partial^2\phi)^2+ \frac{1}{M_{\rm{pl}}^2m^2}\,h\,(\partial^2\phi)(\partial h)(\partial\phi)+ \frac{1}{M_{\rm{pl}}^2m^2}(\partial h)^2(\partial\phi)^2,
\end{aligned}
\end{equation}
where the full schematic expressions are relegated to Appendix \ref{app2}. In these interactions, we see two new energy scales emerging
\begin{equation}
\label{scscale}
    \Lambda_3\equiv (\Lambda_2m)^{1/3},\qquad \Lambda_4\equiv\sqrt{M_{\rm{pl}}m}.
\end{equation}
The first scale $\Lambda_3$ is the familiar strong coupling scale of Generalized Proca theory on flat space that arises in the pure scalar sector \cite{Heisenberg:2014rta,Jimenez:2016isa,Heisenberg:2018vsk,Heisenberg:2020jtr}. The second scale indicates when the Proca field and the graviton become strongly interacting, which can be seen by looking at e.g. the $2\rightarrow 2$ tree-level scattering amplitude $\mathcal{M}_{2\rightarrow 2}\sim \frac{E^4}{M_{\rm{pl}}^2m^2}$ stemming from operators in $\mathcal{L}_2^{(4)}$. Evidently, $\mathcal{M}_{2\rightarrow 2}$ becomes of order one at energies $E\sim \Lambda_4$, beyond which unitarity is violated, and hence $\Lambda_4$ marks the maximal cutoff of the Stückelberg theory. 

Being ignorant of the precise Lorentz structure of the theory, one might have naively expected this scale to lie at energies $\sim (M_{\rm pl}m^2)^{1/3}$, originating from terms $\sim h(\partial^2\phi)^2$, $(\partial h)(\partial^2\phi)(\partial\phi)\subset \mathcal{L}_2^{(3)}$. However, the coefficients of Generalized Proca theory are carefully tuned to avoid ghosts, which leads to the cancellation of such terms with the poorest behaviour. This is in close analogy to massive gravity, where the broken diffeomorphism invariance is restored à la Stückelberg through an additional vector and scalar field. The resulting theory has a strong coupling scale set by higher-order scalar self-interactions, causing the propagation of a ghostly degree of freedom (Boulware-Deser ghost). Explicitly tuning the Lagrangian coefficients, as in Generalized Proca theory, allows to eliminate these interactions, which raises the cutoff to a higher value and renders the theory manifestly ghost-free \cite{Heisenberg:2014rta,Jimenez:2016isa,Hinterbichler:2011tt}. 

To ensure a natural hierarchy of scales, we assume a small classical Proca coupling constant with $m^2\ll \Lambda_2$ such that the scale $\Lambda_3$ lies parametrically far above the mass, serving as an important criterion for the Proca effective field theory to be valid. Further assuming the effective field theory cutoff to obey $\Lambda_2^4\ll m M_{\rm{pl}}^3$, we guarantee that the strong coupling effects arising from interactions with the graviton appear at much higher energies as the ones in the pure Proca sector. This leads to the natural hierarchy 
\begin{equation}
    m \ll \Lambda_3 \ll \Lambda_4 \ll \Lambda_2 \ll M_{\mathrm{pl}},
\end{equation}
where all energy scales are separated by parametrically large gaps, ensuring the healthiness of the Proca effective field theory. To isolate the leading operators at high energies, we take the decoupling limit (DL) to zoom in on the dangerous energy scales at which the theory becomes strongly interacting, while simultaneously, in a controlled way, sending all other scales either to zero or to infinity.
\begin{definition}[Decoupling limit]
\label{def:decoupling}
Let $S[A,g]$ be an action functional belonging to the class defined in
Definition~\ref{def:ghostfreevector}, with parameters $(M_{\rm Pl},m,\Lambda_2)$.
The decoupling limit is defined as the scaling limit
\begin{equation}
    \label{DL}
    m\rightarrow 0,\quad \Lambda_2\rightarrow\infty,\quad \rm{and}\quad  M_{\rm{pl}}\rightarrow\infty\quad \rm{s.t.}\quad 
    \begin{cases} 
        \Lambda_3=(\Lambda_2^2m)^{1/3}=\rm{const.} \\[5pt]
        \Lambda_4=\sqrt{M_{\rm{pl}}m}=\rm{const.}
    \end{cases}\,
\end{equation}
and, the metric and vector fields are expanded around a background solution as in
Definition~\ref{def:perturbativeexpansion}.
\end{definition}
 
This procedure extracts precisely the dominating interactions in the energy band $[\Lambda_3,\Lambda_4]$ and allows to decouple the Goldstone boson from the transverse vector modes, which only survive in their gauge-invariant forms $F_{\delta A}$ and $\tilde{F}_{\delta A}$. Crucially, the simultaneous gauge symmetry \eqref{gaugesymm} breaks apart into two separate transformations 
\begin{equation}
    \phi\rightarrow\phi+ c,\quad \delta A_\mu\rightarrow\delta A_\mu-\partial_\mu\alpha,
    \label{gaugesymmsep}
\end{equation}
where the scalar field only keeps a global shift symmetry that is familiar from Galileon theories. 
\begin{proposition}[Structure of the decoupled action]
\label{prop:decoupledaction}
In the decoupling limit of Definition~\ref{def:decoupling}, the action functional
$S[A,g]$ \eqref{gpaction} admits a well-defined limiting functional $S_{\rm dec}$ that decomposes
into the different sectors governing the high-energy behavior of the theory
\begin{equation}
\begin{aligned}
    &\mathcal{L}_{\rm{DL}}^{(2)}=-\frac{1}{4}F_{\mu\nu}^{\delta A}F_{\delta A}^{\mu\nu}+\frac{1}{2}(\partial\phi)^2+\frac{1}{2}h^{\mu\nu}(\partial^2h_{\mu\nu})-\frac{1}{4}h^\mu{}_\mu(\partial^2h^\mu{}_\mu),\\[10pt]
    &\mathcal{L}^{(3)}_{\rm{DL}}=\frac{1}{\Lambda_3^3\,}c_3\,(\partial\phi)^2(\partial^2\phi)+\frac{2}{\Lambda_4^2}h_{\beta}{}^\gamma F_{\delta A}^{\alpha\beta}(\partial_\alpha\partial_\gamma\phi)+\frac{2}{\Lambda_4^2}(\partial_\beta h_{\alpha\gamma})F_{\delta A}^{\beta\alpha}(\partial^\gamma\phi),\\[10pt]
    &\mathcal{L}^{(4)}_{\rm{DL}}\sim \frac{1}{\Lambda_4^4}\,h^2\,(\partial^2\phi)^2+ \frac{1}{\Lambda_4^4}\,h\,(\partial^2\phi)(\partial h)(\partial\phi)+ \frac{1}{\Lambda_4^4}(\partial h)^2(\partial\phi)^2,
    \label{LagrDL}
\end{aligned}
\end{equation}
where the explicit Lorentz contractions in $\mathcal{L}^{(4)}_{\rm{DL}}$ are omitted for brevity, see Appendix \ref{app2} for more details. 
\end{proposition}

Eq. \eqref{LagrDL} reveals surprisingly simple structure of the theory in the decoupling limit: While the Stückelberg-ed Lagrangians \eqref{LpreDL} (see Appendix \ref{app2} for the full expressions) contain numerous terms that mix the graviton with the scalar and vector sector, where particularly the vector degree of freedom do not necessarily appear in gauge-invariant form, the resulting theory in the decoupling limit, however, collapses to a small, tractable number of classical interactions, as only few terms from \eqref{LpreDL} survive. Furthermore, the decoupling limit in Eq. \eqref{LagrDL} reflects the anticipated high-energy behaviour of the theory: the longitudinal polarization is manifestly broken apart from the transverse modes and propagates independently. This is in accordance with the Goldstone boson equivalence theorem, stating that all three vector polarizations are indistinguishable at low energies, but at high energies, the dynamics of the Goldstone boson becomes distinct, as it decouples from the two transverse polarizations. Regarding the cubic Lagrangian, the pure Proca term is clearly governed by the longitudinal mode and reduced to its Galileon form. All other terms describe interactions among the scalar, vector and tensor sector, where the vector degrees of freedom manifestly enter in gauge-invariant form only and are hence protected from quantum detuning and ghost instabilities by the $U(1)$ gauge symmetry \eqref{gaugesymmsep}. Therefore, the decoupling limit provides a useful tool to assess the stability of an effective field theory under radiative corrections, as it isolates the physics right at the dangerous energy scales and reduces the analysis to the behaviour of the Goldstone boson only. 

 \begin{lemma}[Decoupling Limit Control of Radiative Corrections]\label{lem:DL}
Any local operator generated by quantum corrections that contributes to physical amplitudes at energies $E \ll \Lambda_{\text{cutoff}}$ must survive in the decoupling limit with a finite coefficient. Conversely, any operator that vanishes or becomes parametrically suppressed in the decoupling limit cannot be generated radiatively with a coefficient unsuppressed by the effective field theory cutoff.
In particular, the decoupling limit provides a complete characterization of all radiatively admissible counterterms relevant below the cutoff, and the absence of a given operator in this limit implies its non-renormalization within the effective field theory expansion.
\end{lemma}

As a consequence, the analysis of quantum corrections in the decoupling limit is sufficient to determine the radiative stability of the full theory to all loop orders within its regime of validity.



\section{Quantum Corrections}\label{quantumcorrections}

\subsection{Two-point functions}
\label{sec51}
Before presenting explicit loop computations, we emphasize that the purpose of the following sections is not a diagram-by-diagram analysis of quantum corrections, but rather a systematic classification of the local operators that can be generated radiatively within the effective field theory. The individual Feynman diagrams computed below serve as representative probes of the operator algebra of the theory and are used to identify the derivative structure and scale suppression of admissible counterterms. In particular, our focus is on determining whether quantum corrections can generate operators with derivatives per field as those appearing in the classical ghost-free interactions. As we demonstrate, this never occurs: all radiatively induced operators are necessarily higher-derivative and suppressed by the strong-coupling scales identified in the decoupling limit. The explicit loop calculations presented in this section therefore provide constructive confirmation of a general operator-level statement, rather than isolated perturbative results.

Before computing one-loop corrections explicitly, it is instructive to perform a power counting analysis of the resulting one-loop diagrams in order to see the types of quantum-induced operators that can appear and potentially detune the classical interactions. Given that there is no symmetry protecting Generalized Proca theory from receiving large quantum corrections, we explicitly expect such dangerous terms to arise at the quantum level. The main danger lies in the fact that such operators come with higher derivatives, which could lead to ghost instabilities within the effective field theory's regime of validity. At this stage, we choose to work in unitary gauge and hence focus on Generalized Proca theory in the original formulation, with the kinetic term \eqref{Lkin} and classical interactions \eqref{Lcubic}-\eqref{Lquartic}. Since the vector propagator goes like $\sim \frac{1}{m^2}$ in the old formulation, we anticipate counterterms enhanced by inverse powers of the mass that prevail over the classical interactions at high energies, which would destabilize the classical effective field theory structure. This behaviour becomes very evident in the decoupling limit, where such quantum-enhanced operators diverge explicitly as we take $m\to 0$. Nevertheless, as we will show in Sec. \ref{sec52} for one-loop corrections to the two--point functions, all these terms either contain the vector field in gauge-invariant form, which softens their behaviour in the decoupling limit, or they simply do not appear and hence must experience non-trivial cancellations in unitary gauge. In Sec. \ref{sec53}, we provide a natural explanation for these cancellations by constructing the resulting one-loop counterterms directly from the decoupling limit. Overall, we observe that the non-renormalizable operators generated from one-loop quantum corrections respect the structure and hierarchy of the classical effective field theory and therefore conclude radiative stability of Generalized Proca theories.

\subsubsection{Power-counting estimates}

We now examine the two-point diagrams at one-loop that can be constructed from the cubic and quartic interaction vertices of the Lagrangians \eqref{Lcubic}-\eqref{Lquartic}. These diagrams induce quantum corrections to the quadratic part of the action and could potentially detune the kinetic structure of the theory. For now, we restrict the analysis to the lowest-order corrections involving only one graviton, since higher orders in $h$ are further suppressed by inverse powers of the Planck mass $\sim 1/M_{\rm{pl}}^n$ with $n>2$. Employing standard Feynman diagram techniques and working in the $\overline{\rm{MS}}$-scheme of perturbative renormalization, the dangerous destabilizing terms arise from the UV-divergent parts of the one-particle irreducible (1PI) diagrams, which we first estimate by power-counting before giving exact results later in Sec. \ref{sec52}. We restrict our analysis to 1PI diagrams that involve mixings between the vector field and the graviton, since, on the one hand, the pure Proca sector on flat space was already studied in \cite{Heisenberg:2020jtr} and shown to be quantum-stable, and, on the other hand, pure gravity is known to be perturbatively non-renormalizable \cite{Donoghue:1994dn}. Overall, these mixing diagrams can be divided into three physically distinct classes: $(i)$ corrections to the Proca propagator, $(ii)$ corrections to the graviton propagator, and $(iii)$ mixed vector–graviton two-point functions that induce kinetic mixing between the two sectors.\\

\paragraph{$(i)$ Proca propagator corrections:}
Corrections to the Proca two-point function arise from two distinct one-loop 1PI diagrams that include vector-tensor mixing, consisting of a tadpole and a bubble contribution depticted in Fig. \ref{fig2pt}.

\begin{figure}[H]
    \centering
    \includegraphics[width=0.7\linewidth]{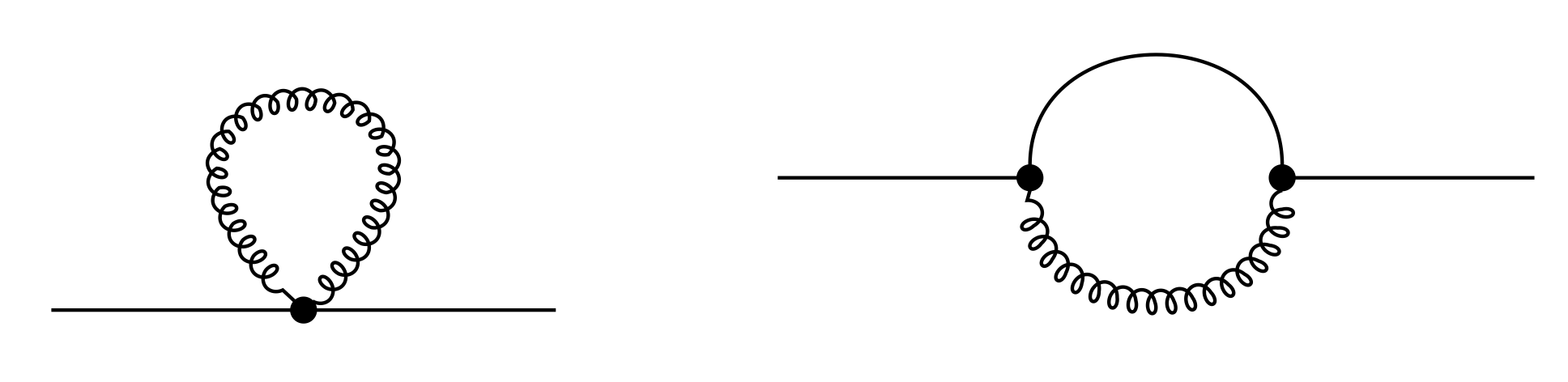}
    \caption{\small{The two mixed one-loop 1PI diagrams inducing quantum corrections to the Proca two-point function, both accompanied by a symmetry factor of~2. The vertices required for these diagrams arise from the perturbative expansion of the Lagrangian in the minimal model at cubic and quartic order that involve mixed vector--tensor terms, i.e.\ from $\mathcal{L}^{(3)}_{h\,\delta A^2}\subset\mathcal{L}^{(3)}$ and $\mathcal{L}^{(4)}_{h^2\delta A^2}\subset\mathcal{L}^{(4)}$.}}
    \label{fig2pt}
\end{figure}
The power counting estimates include all terms potentially arising in dimensional regularization that are allowed by Lorentz invariance. For the above diagrams from Fig. \ref{fig2pt}, we expect the following structures
\begin{equation}
\begin{aligned}
\label{pcprocaprop}
    &\mathcal{M}_{\mathcal{L}^{(4)}_{h^2\delta A^2}}\sim\frac{m^{2}}{M_{\rm Pl}^{2}}\!\left(m^{2} + p^{2} + \frac{p^{4}}{m^{2}}\right),\\[5pt]
    &\mathcal{M}_{\mathcal{L}^{(3)}_{h\,\delta A^2},\mathcal{L}^{(3)}_{h\,\delta A^2}}\sim\frac{m^{2}}{M_{\rm Pl}^{2}}\!\left( m^{2} + p^{2} + \frac{p^{4}}{m^{2}} + \frac{p^{6}}{m^{4}} \right), 
\end{aligned}
\end{equation}
where the terms with inverse mass powers reflect the number of internal propagators in the loop, which is maximally two. The presence of momentum powers indicates that these diagrams can generate irrelevant operators with two or more derivatives acting on the vector field: For instance, terms $\sim p^2$ in \eqref{pcprocaprop} induce counterterms of the following schematic form
\begin{equation}
\label{ghost}
    \sim \frac{m^2}{M_{\rm{pl}}^2}\,(\partial\,\delta A)^2,
\end{equation}
that could potentially detune the kinetic term of the vector field. Yet, at a closer look, these terms are parametrically suppressed by the effective field theory hierarchy $m\ll M_{\rm{pl}}$, with the mass of the induced ghost of order $\sim \Lambda_4^4/m^2$ (after canonical normalization) lying far beyond the effective field theory cutoff, such that they do not threaten the classical effective field theory structure. However, terms involving higher powers in momenta come with additional risks, as they have the mass appearing explicitly in their denominator. Looking at the highest-momentum contributions $\sim p^4$ and $\sim p^6$ in \eqref{pcprocaprop}, we anticipate the following quantum operators in the decoupling limit
\begin{equation}
    \begin{aligned}
        &\sim\frac{\partial^4}{M_{\rm{pl}}^2}\,\delta A^2\xrightarrow[]{\text{DL}}\frac{\partial^4}{M_{\rm{pl}}^2m^2}\,(\partial\phi)^2=\frac{\partial^4}{\Lambda_4^4}(\partial\phi)^2,\\[5pt]
        &\sim \frac{\partial^6}{M_{\rm{pl}}^2m^2}\,\delta A^2\xrightarrow[]{\text{DL}}\frac{\partial^6}{\Lambda_4^4m^2}\,(\partial\phi)^2,
    \end{aligned}
\end{equation}
where we have focused on the pure scalar sector when applying the Stückelberg transformation \eqref{stb}, as these terms exhibit the most pathological behaviour at high energies. While terms $\sim p^4 $ still induce counterterms that remain safely suppressed in the decoupling limit, as the inverse mass powers are absorbed by the strong coupling scale, terms $\sim p^6$ lead to operators that naively blow up in the decoupling limit, such that their destabilizing nature becomes evident. Therefore, when doing explicit one-loop computations, we have to pay special attention to corrections coming from the bubble diagram from Fig. \ref{fig2pt} and verify that these counterterms do not appear.\\

\paragraph{$(ii)$ Corrections to the graviton propagator: } For the graviton two-point function, we expect one-loop corrections stemming from similar tadpole and bubble 1PI diagrams as in case $(i)$ that are given in Fig. \ref{fig2ptgrav}. 

\begin{figure}[H]
    \centering
    \includegraphics[width=0.7\linewidth]{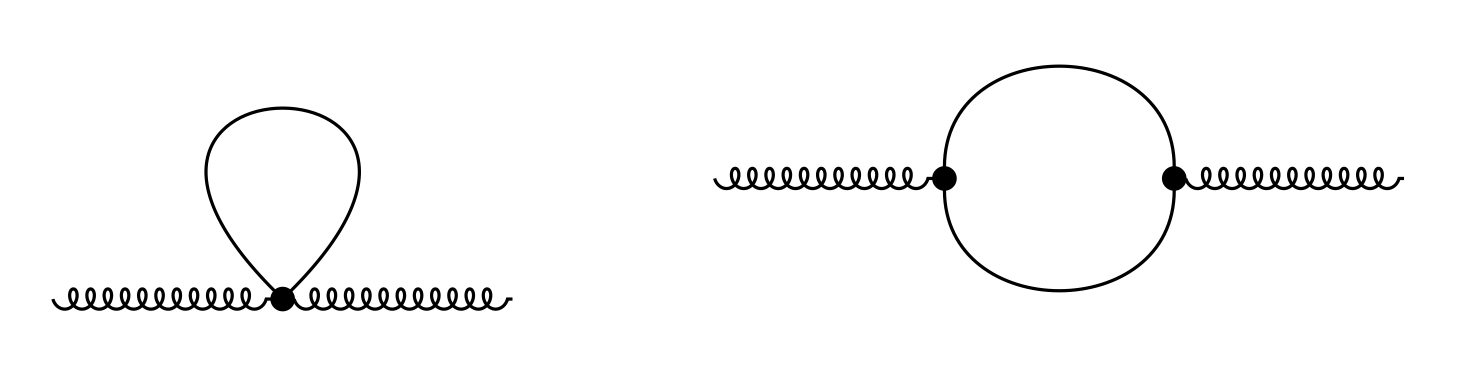}
    \caption{\small{The two mixed one-loop 1PI diagrams inducing quantum corrections to the graviton two-point function, both accompanied by a symmetry factor of 2. The vertices required for these diagrams arise from the perturbative expansion of the Lagrangian in the minimal model at cubic and quartic order that involve mixed vector-tensor terms, i.e. from $\mathcal{L}^{(3)}_{h\,\delta A^2}\subset\mathcal{L}^{(3)}$ and $\mathcal{L}^{(4)}_{h^2\delta A^2}\subset \mathcal{L}^{(4)}$. }}
    \label{fig2ptgrav}
\end{figure}
Based on power counting in dimensional regularization together with Lorentz invariance, we anticipate the following results
\begin{equation}
\label{pcgravprop}
    \begin{aligned}
        &\mathcal{M}_{\mathcal{L}^{(4)}_{h^2\delta A^2}}\sim\frac{m^{2}}{M_{\rm Pl}^{2}}\!\left( m^{2} + p^{2} + \frac{p^{4}}{m^{2}} + \frac{p^{6}}{m^{4}} \right),\\[5pt]
        &\mathcal{M}_{\mathcal{L}^{(3)}_{h\,\delta A^2},\mathcal{L}^{(3)}_{h\,\delta A^2}}\sim\frac{m^{2}}{M_{\rm Pl}^{2}}\!\left( m^{2} + p^{2} + \frac{p^{4}}{m^{2}} + \frac{p^{6}}{m^{4}}+\frac{p^8}{m^6} \right).
    \end{aligned}
\end{equation}
Focusing on the highest-momentum terms that are the most dangerous in the decoupling limit, we obtain quantum-induced counterterms of the form
\begin{equation}
\label{ct2ptgrav}
    \begin{aligned}
        &\sim\frac{\partial^6}{M_{\rm Pl}^{2}m^2}\,h^2=\frac{\partial^6}{\Lambda_4^4}\,h^2,\\[5pt]
        &\sim\frac{\partial^8}{M_{\rm Pl}^{2}m^4}\,h^2=\frac{\partial^8}{\Lambda_4^4m^2}\,h^2,
    \end{aligned}
\end{equation}
where particularly the second counterterm in Eq. \eqref{ct2ptgrav} stemming from the bubble diagram in Fig. \ref{fig2ptgrav} diverges in the decoupling limit, leading to a breakdown of the classical effective field theory structure. However, when computing the associated one-loop corrections explicitly in Sec. \ref{sec6}, we will observe that the Proca effective field theory is structured in a way that does not allow for such destabilizing operators to appear and all potentially worrisome terms are viably suppressed by the strong coupling scale $\Lambda_4$ without additional enhancement from inverse mass powers.\\ 

\paragraph{$(iii)$ Mixed Proca-graviton propagators:  }
The remaining two one-loop 1PI diagrams displayed in Fig. \ref{fig2ptmix} that involve both vector and tensor sectors contribute to two-point functions with one external Proca and one external graviton leg, hence inducing kinetic mixing between \(A_\mu\) and \(h_{\mu\nu}\) at the quantum level. As such mixed propagators do not appear classically in our chosen background configuration \eqref{bgconfig}, verifying that such mixing terms are heavily suppressed at high energies is crucial for the stability of the Proca effective field theory. 
\begin{figure}[H]
    \centering
    \includegraphics[width=0.7\linewidth]{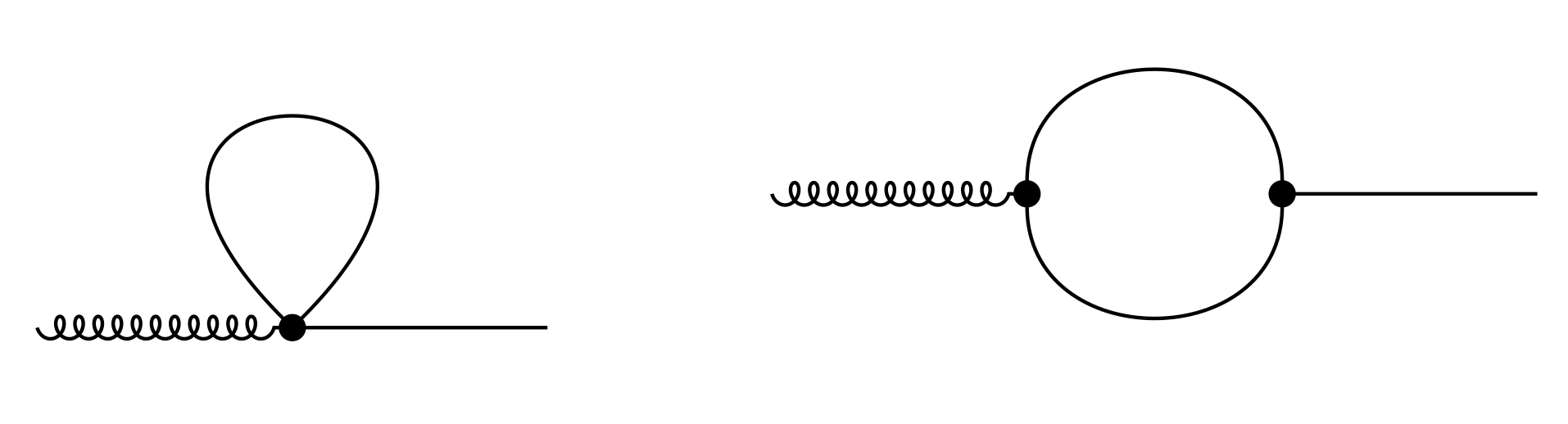}
    \caption{\small{The two mixed one-loop 1PI diagrams inducing kinetic mixing at the quantum level, both accompanied by a symmetry factor of 2. The vertices required for these diagrams arise from the perturbative expansion of the Lagrangian in the minimal model at cubic and quartic order that involve mixed vector-tensor and pure Proca terms, i.e. from $\mathcal{L}^{(3)}_{h\,\delta A^2},\,\mathcal{L}^{(3)}_{\delta A^3}\subset\mathcal{L}^{(3)}$ and $\mathcal{L}^{(4)}_{h\delta A^3}\subset \mathcal{L}^{(4)}$. }}
    \label{fig2ptmix}
\end{figure}
Power counting and Lorentz invariance imply corrections of the form
\begin{equation}
\label{pcmixedprop}
    \begin{aligned}
        &\mathcal{M}_{\mathcal{L}^{(4)}_{h\delta A^3}}\sim\frac{m^{4}}{\Lambda_2^{2} M_{\rm Pl}}\!\left( p + \frac{p^{3}}{m^{2}} + \frac{p^{5}}{m^{4}} \right),\\[5pt]
        &\mathcal{M}_{\mathcal{L}^{(3)}_{h\,\delta A^2},\mathcal{L}^{(3)}_{\delta A^3}}\sim\frac{m^{4}}{\Lambda_2^{2} M_{\rm Pl}}\!\left( p + \frac{p^{3}}{m^{2}} + \frac{p^{5}}{m^{4}} + \frac{p^{7}}{m^{6}} \right).
    \end{aligned}
\end{equation}
Here, particularly the contribution $\sim p^7$ can become dangerous in the decoupling limit, where it leads to the following divergent operator
\begin{equation}
    \sim \frac{\partial^7}{\Lambda_2^2M_{\rm{pl}}m^2}\,(h\,\delta A)\xrightarrow[]{\text{DL}}\frac{\partial^7}{\Lambda_3^3 \Lambda_4^2\, m}\,(h\,\partial\phi).
\end{equation}
Again, the precise Feynman calculation of the one-loop corrections will reveal that such pathological terms are either tamed by the vector appearing in gauge-invariant form or absent altogether. Having successfully classified all possible 1PI diagrams at one loop and identified the destabilizing quantum-induced operators arising from the highest momentum terms in the bubble diagrams, we now turn to computing these one-loop corrections explicitly.

\subsubsection{Explicit Feynman diagram calculation}
\label{sec52}
Using standard Feynman calculus, we compute the one-loop quantum corrections to Generalized Proca theory in unitary gauge. These corrections stem from the one-loop 1PI graphs that all give contributions to the reduced matrix element $\mathcal{M}_{\rm{1PI}}$ in the perturbative treatment of the scattering matrix. In perturbative renormalization within the $\overline{\rm{MS}}$-scheme, concrete counterterms to the classical theory are invoked to precisely cancel the UV-divergences generated by loop diagrams
in dimensional regularization, which is carried out setting $d=4-2\varepsilon$. The UV-divergent part of a generic one-loop diagram with $N$ external legs is denoted by $\mathcal{M}_N^{\rm{div}}$ and depends on $N-1$ external momenta $p_1,...,p_{N-1}$, where the  remaining $p_N$ is fixed by momentum conservation. From now on, we assume all external momenta to be incoming. Following these steps, we obtain the contribution to the UV-divergent piece of the reduced matrix element coming from the diagrams in Fig. \ref{fig2pt}
\begin{equation}
\label{M2divproca}
    \begin{aligned}
        \mathcal{M}_{2,\,\delta A}^{\rm{div}}\supset \frac{\epsilon_{p}^\mu\epsilon^\nu_{-p}}{48\,\pi^2\,\varepsilon\,M_{\rm Pl}^{2}}
        \big[
        \big(3 p^{4} + 26 p^{2} m^{2} - 9 m^{4}\big) \eta_{\mu\nu}-\big(3 p^{2} + 23 m^{2}\big) p_{\mu} p_{\nu}\big]  ,
    \end{aligned}
\end{equation}
where the external momenta are on-shell such that $p\equiv p_1=-p_2$, and $\epsilon_{p}^\alpha$ denotes the vector polarization associated with external momentum $p$. Clearly, the resulting expression inherently fits with the prediction \eqref{pcprocaprop} based on dimensional analysis, however, the dangerous $\sim p^8$ contribution is manifestly absent. Furthermore, the highest-momentum contribution in \eqref{M2divproca} appears in the gauge-invariant form $\eta_{\mu\nu}\partial^2-\partial_\mu \partial_\nu$, which tames the high-energy behaviour of the induced counterterm 
\begin{equation}
    \sim \frac{\partial^2}{M_{\rm{pl}}^2}\,F_{\delta A}^2\xrightarrow[]{\text{DL}}0,
\end{equation}
even to the point where the latter vanishes in the decoupling limit, such that corrections to the gauge-invariant sector are only expected to enter at the Planck scale. Therefore, the structure of Generalized Proca theory is organized in a healthy way that leads to non-trivial cancellations of the higher-momentum terms anticipated in unitary gauge based on dimensional grounds and symmetry considerations, which prevents the appearance of destabilizing quantum-induced counterterms at the two-point level. We will justify this behaviour more rigorously in the ensuing Sec. \ref{sec53} based on the theory's structure in the decoupling limit.

Next, we turn our attention towards the one-loop corrections to the graviton propagator. Repeating the same steps as above, we extract the UV-divergent part of the reduced 1PI matrix element induced by the graphs in Fig. \ref{fig2ptgrav}

\begin{equation}
\label{M2divgrav}
    \mathcal{M}_{2,h}^{\rm{div}}\supset \frac{\epsilon^{\mu\nu}_p\,\epsilon^{\alpha\beta}_{-p}}{1920\, \pi^2\varepsilon \,M_{\rm{pl}}^2}
    \Big[m^4M_{\mu\nu\alpha\beta}^{(0)}(p)+m^2M_{\mu\nu\alpha\beta}^{(2)}(p)+M_{\mu\nu\alpha\beta}^{(4)}(p)+\frac{1}{m^2}M_{\mu\nu\alpha\beta}^{(6)}(p)\Big],
\end{equation}
where we have abbreviated the Lorentz structures $M^{(n)}_{\mu\nu\alpha\beta}$ by tensors collecting all momentum powers of order $n$ to avoid clutter (see Appendix \ref{app3} for the full answer). The metric polarizations are denoted by $\epsilon_{\mu\nu}$ and we have again put the external momenta on-shell $p\equiv p_1=p_2$. Notably, the highest momentum contribution $\sim p^8$ anticipated in power counting \eqref{pcgravprop} that leads to divergences in the decoupling limit is not present in \eqref{M2divgrav}, proving that also corrections to the graviton two-point function induced by mixed vector-tensor loops are safely suppressed at high energies. Again, in the exact expression \eqref{M2divgrav}, we hence witness non-trivial cancellations of all troublesome terms that were a priori allowed by dimensional analysis in unitary gauge. Further comments on this fortunate behaviour are given in Sec. \ref{sec53}.

Finally, we focus on the one-loop diagrams in Fig. \ref{fig2ptmix} that lead to mixed two-point functions between the Proca field and the graviton. The explicit computation of the UV-divergence in the associated reduced matrix element yields 
\begin{equation}
    \mathcal{M}_{2,\,h\delta A}^{\rm{div}}\supset \frac{i\, c_3\,\epsilon_p^{{\mu\nu}}\epsilon_{-p}^\alpha}{192\, \pi ^2 \varepsilon\,\Lambda_2 ^2\, M_{\rm{pl}}}
    \big[3\, m^4 p_{\mu } \eta_{\alpha \nu }-3\, m^4 p_{\nu } \eta_{\alpha \mu }+\left(6\, m^2-p^{2}\right)
   p_{\alpha } \left(p_{\mu } p_{\nu }-p^{2} \eta_{\mu \nu }\right)\big],
\end{equation}
which does not contain the dangerous momentum power $\sim p^7$ allowed in the power counting estimate \eqref{pcmixedprop}. Therefore, the induced quantum counterterms all exhibit a healthy behaviour in the high-energy limit, as they are heavily suppressed by the strong coupling scales of the theory \eqref{scscale} and hence do not destroy the classical kinetic structure that decomposes into an independent Proca and graviton sector. In the next section, we argue based on the theory's structure in the decoupling limit why such dangerous terms do not appear in the explicit one-loop contributions to the reduced matrix elements, which puts these seemingly magical cancellations on firm theoretical grounds.

\subsubsection{Decoupling limit analysis}
\label{sec53}
In the previous section, we observed non-trivial cancellations in unitary gauge of the highest-momenta terms in the reduced matrix elements, which saved Generalized Proca theory from incurring destabilizing counterterms at the quantum level. As a next step, we show that this behaviour can easily be explained from the decoupling limit: Looking at $\mathcal{L}^{(3)}_{\rm{DL}}$ in \eqref{LagrDL}, we notice that the only remaining cubic interactions among different sectors couple the graviton $h$ and the scalar field equipped with additional derivatives $\sim \partial\phi,\,\partial^2\phi$ to the gauge-invariant vector $F_{\delta A}$, which tightly constrains the structure of the bubble diagram to the forms of Fig. \ref{figDLbubble}. In addition, the structure of $\mathcal{L}_{\rm{DL}}^{(4)}$ in \eqref{LagrDL} allows for interactions between two gravitons and two scalars, such that there is a remaining tadpole contribution in the decoupling limit depicted in Fig. \ref{figDLtadpole}.

\begin{figure}[H]
    \centering
    \includegraphics[width=0.7\linewidth]{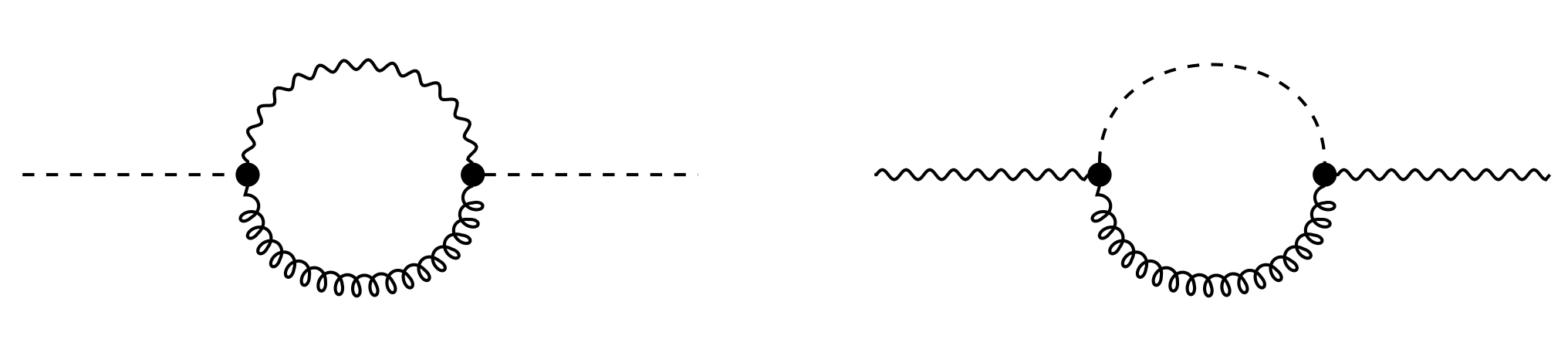}
    \caption{\small{The two mixed-sector one-loop 1PI bubble diagrams inducing quantum corrections to the Proca two-point function, when constructed from the decoupling limit. The vertices required for these diagrams arise from the scalar-vector-tensor interactions in $\mathcal{L}^{(3)}_{\rm{DL}}$. Dashed lines denote external legs or propagators of the scalar field $\phi$ that couples to the graviton and the gauge-invariant vector field $F_{\delta A}$, whose propagators and external legs are depicted by wavy lines. }}
\label{figDLbubble}
\end{figure}

\begin{figure}[H]
    \centering
    \includegraphics[width=0.7\linewidth]{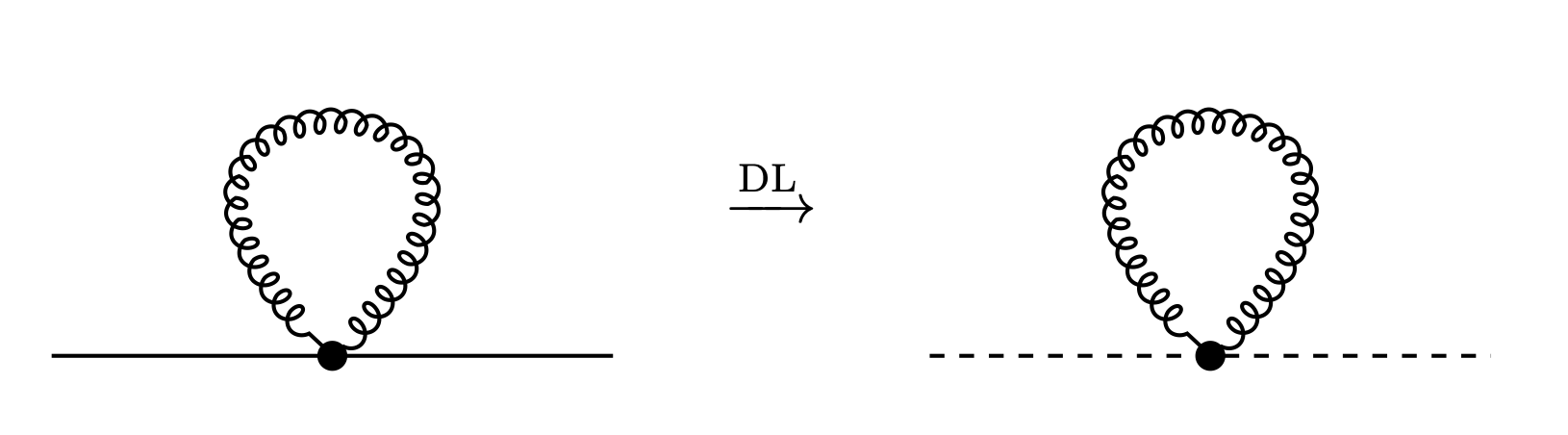}
    \caption{\small{The tadpole diagram in the decoupling limit, giving quantum corrections to the Proca two-point function. The left diagram corresponds to the one-loop contribution in unitary gauge, whereas the right one represents the surviving configuration in the decoupling limit, with dashed lines indicating external scalar legs.}}
\label{figDLtadpole}
\end{figure}
In the decoupling limit, all propagators behave as $\sim 1/p^2$ at high energies and vertices are suppressed by factors of the strong coupling scales $\sim 1/\Lambda_3^3$ or $\sim 1/\Lambda_4^2$, while the mass scale  $m\to 0$ is sent to zero. Based on dimensional analysis, the diagrams from Figs. \ref{figDLbubble} and \ref{figDLtadpole} hence induce the following counterterms
\begin{equation}
\label{ctprocakin}
    \mathcal{L}_{\rm{DL},\,ct}^{(2),\,\delta A^2}\sim \frac{\partial^4}{\Lambda_4^4}(\partial\phi)^2
    +\frac{\partial^4}{\Lambda_4^4}F_{\delta A}^2,
\end{equation}
which all are manifestly dressed by more derivatives than the kinetic term of the vector field such that they do not contribute to its renormalization. Crucially, we observe that the worrisome $\sim p^6$ term from power counting \eqref{pcprocaprop} is absent, because it simply cannot be formed in the decoupling limit, which explains the non-trivial cancellations appearing in unitary gauge that ended up curing the effective field theory structure. 

Next, we focus on the corrections induced on the graviton two-point function at one loop in the decoupling limit. The resulting diagrams are again constructed from the surviving vertices in $\mathcal{L}_{\rm{DL}}^{(3)}$ and $\mathcal{L}_{\rm{DL}}^{(4)}$ from \eqref{LagrDL} and depicted in Fig. \ref{figDLbubblegrav}.

\begin{figure}[H]
    \centering
    \includegraphics[width=0.7\linewidth]{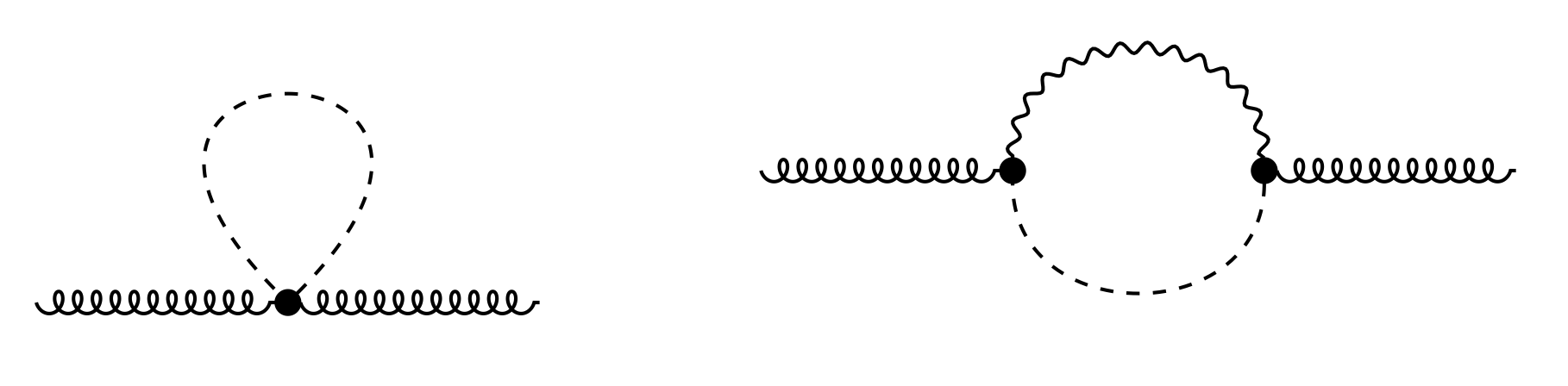}
    \caption{\small{The two mixed-sector one-loop 1PI diagrams inducing quantum corrections to the graviton two-point function that prevail in the decoupling limit. The vertices required for these diagrams arise from $\mathcal{L}^{(3)}_{\rm{DL}}$ and $\mathcal{L}^{(4)}_{\rm{DL}}$, where dashed lines correspond to scalar legs or propagators, and wavy lines to gauge-invariant vector legs or propagators. }}
\label{figDLbubblegrav}
\end{figure}
These diagrams lead to the following irrelevant operator arising at the quantum level
\begin{equation}
\label{ctgravkin}
    \mathcal{L}_{\rm{DL},\, ct}^{(2),\,h^2}\sim \frac{\partial^4}{\Lambda_4^4}(\partial h)^2
\end{equation}
where the dangerous high-momentum terms $\sim p^8$ anticipated in power counting do not appear, which manifestly saves the graviton two-point function from destabilizing counterterms stemming from scalar-vector-tensor interactions at high energies. 

Finally, we examine the mixed one-loop diagrams from Fig. \eqref{fig2ptmix} in the decoupling limit. From the surviving interactions in \eqref{LagrDL}, we directly observe that these graphs with an external graviton and Proca leg each cannot be assembled at all, meaning that the classical kinetic structure remains manifestly protected from such kinetic mixing at the quantum level. 

Overall, the above power counting analysis in the decoupling limit provided a natural explanation for the absence of the highest-power momentum terms anticipated in unitary gauge that would have a destabilizing effect on the Generalized Proca effective field theory. Due to the limited number of remaining interactions \eqref{LagrDL}, only very few one-loop diagrams survive at high energies in the first place, which produce counterterms that are all safely suppressed by the strong coupling scale of the Proca-graviton interactions. While these counterterms naturally do introduce ghosts owing to their higher derivative structure, the associated ghostly masses, however, lie safely above the Proca effective field theory cutoff, such that we do not incur into ghost instabilities within the effective field theory's regime of validity. Overall, all two-point operators generated at the quantum level from one-loop 1PI graphs preserve the classical kinetic structures and respect the crucial effective field theory hierarchy between classical and quantum terms.




\subsection{Three-point functions}
\label{sec6}
We now turn to one-loop diagrams that contribute to the perturbative renormalization of the cubic sector of the theory, by generating local counterterms with three external fields. These diagrams can naturally be organized into four distinct classes according to the field content of their external legs, which allows for corrections to the classical vertices $\sim \delta A^3$, $\sim \delta A^2h$, $\sim h^2\delta A$ and $\sim h^3$ appearing in the full theory \eqref{Lcubic}, where we have suppressed the specific derivative structure for brevity. In the previous section, we observed that only very few one-loop graphs prevail in the decoupling limit based on the comparatively simple structure the theory admits at these energy scales. As the number of contributing one-loop diagrams drastically increases at the three-point level compared to two points, we choose to directly focus on the relevant graphs that actually survive at high energies, as their induced counterterms pose the largest threat to the classical effective field theory structure.

\subsubsection{One-loop corrections to the cubic Proca vertex}
The first set of diagrams carries three external vector legs and thus contributes to the renormalization of cubic interactions of the schematic form $\sim \delta A^3$. Given the interactions \eqref{LagrDL} in the decoupling limit, there is only possible one-loop diagram displayed in Fig. \ref{fig3ptAAA} yielding a correction to the cubic pure Proca vertex at high energies.

\begin{figure}[H]
    \centering
    \includegraphics[width=0.7\linewidth]{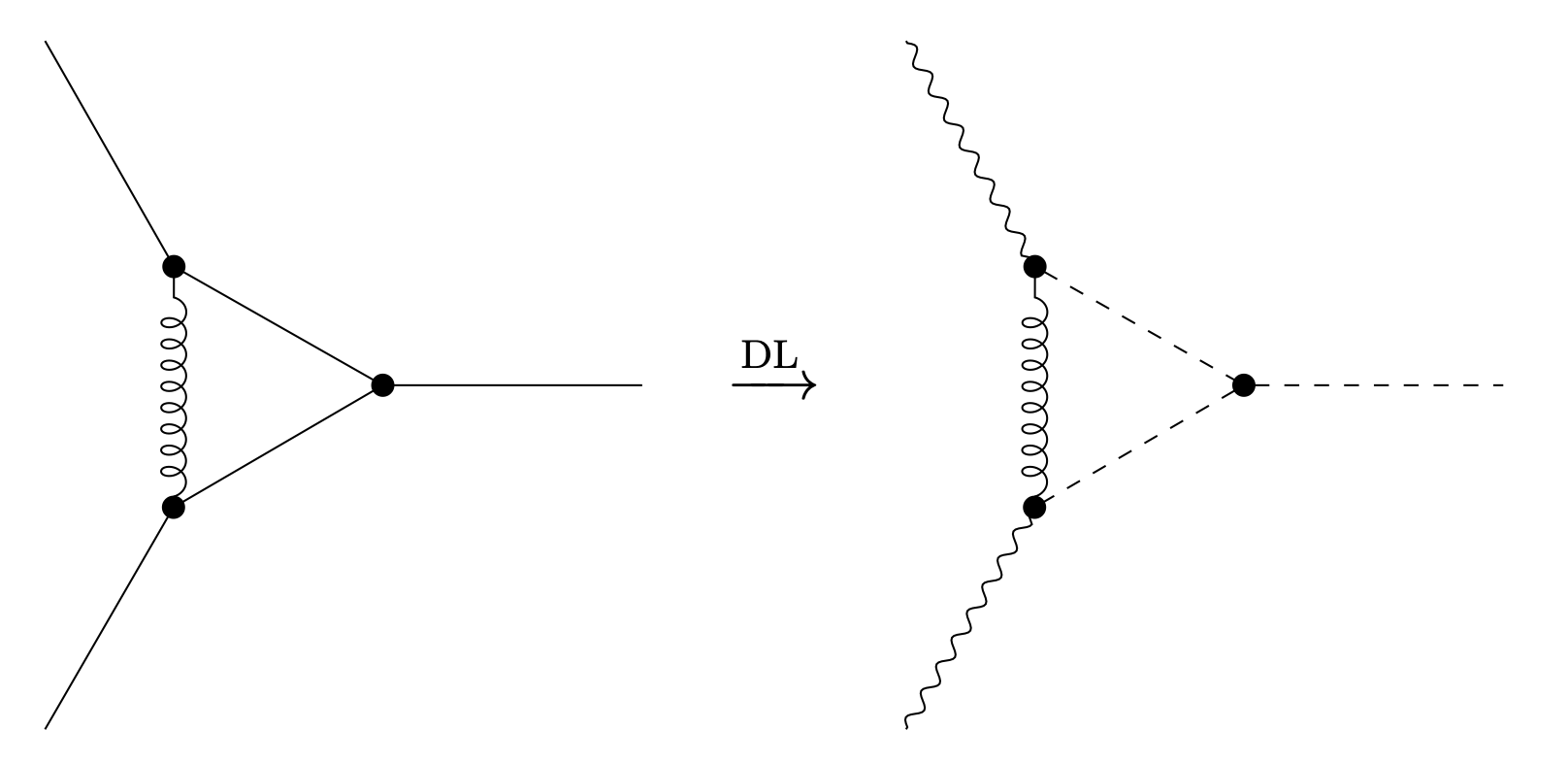}
\caption{\small{The only one-loop 1PI diagram prevailing in the decoupling limit that induces quantum corrections to the classical three-point vertex $\sim \delta A^3$. The left diagram corresponds to the one-loop contribution in unitary gauge, whereas the right one represents the surviving configuration in the decoupling limit, with wavy lines denoting gauge-preserving vector legs $\sim F_{\delta A}$ and dashed lines scalar legs or propagators.}}
\label{fig3ptAAA}
\end{figure}

Power counting the left-hand diagram in Fig. \ref{fig3ptAAA} in unitary gauge within dimensional regularization gives the following naive estimate 

\begin{equation}
    \mathcal{M}_{(\mathcal{L}_{2}^{(3)})^2,\,\mathcal{L}_{3}^{(3)}}\sim \frac{m^2}{\Lambda_2^2 M_{\rm pl}^2}\left(m^2 p+p^3+\frac{p^5}{m^2}+\frac{p^7}{m^4}\right),
\end{equation}
where $p$ indicates external momenta and the series stops at $\sim 1/m^4$ which accounts for two internal vector propagators. Following the same logic as in the previous subsection, we need to be careful about the highest momentum contribution $\sim p^7$, as it amplifies the size of the correspondingly induced counterterm in the decoupling limit and can hence destabilize the effective field theory structure. 

Explicitly calculating the UV-divergent piece of the associated off-shell reduced matrix element formally yields 
\begin{equation}
\begin{aligned}
    \mathcal{M}^{\rm{div}}_{3,\,\delta A^3}\supset\, & \frac{m^4 }{96 \pi^2 \varepsilon M_{\rm pl}^2\Lambda_2^2}\left[M_3^{(1)}(p_1,p_2)+\frac{1}{m^2}M_3^{(3)}(p_1,p_2)+\frac{1}{m^4}M_3^{(5)}(p_1,p_2)\right.\\[10pt]
    &\left.+\frac{1}{m^6}M_3^{(7)}(p_1,p_2)\right],
\end{aligned}
\end{equation}
where we abbreviated all possible contributions from external momenta $p_1,\,p_2$ at power $i$ with $M_3^{(i)}(p_1,p_2)$ due to the increasing complexity of the expressions. We give the explicit results for $M_3^{(i)}(p_1,p_2)$ in an ancillary file and just state the leading piece here

\begin{equation}
\label{M3AAA}
    M_3^{(1)}(p_1,p_2)=126 \left[\epsilon_{12} (\epsilon p_{31}+\epsilon p_{32})-\epsilon_{13} \epsilon p_{22}-\epsilon_{23} \epsilon p_{11}\right],
\end{equation}
where we used the short-hand notations $\epsilon_{ij}\equiv \epsilon_{p_i}\cdot\epsilon_{p_j}$, $\epsilon p_{ij}\equiv \epsilon_{p_i}\cdot p_j$. Evidently, the dangerous contribution $\sim p^7$ is present in Eq. \eqref{M3AAA}, which is worrisome at first glance. However, dimensional analysis of the right-hand diagram in Fig. \ref{fig3ptAAA} in the decoupling limit reveals that this contribution comes in manifestly gauge-invariant form: in the decoupling limit, all propagators admit the healthy form $\sim 1/p^2$ and vertices come with factors $\sim 1/\Lambda_3^3$ or $\sim 1/\Lambda_4^2$, depending on whether they stem from pure Proca or Proca-graviton interactions. Invoking power counting in dimensional regularization together with Lorentz invariance, we find the following quantum-induced counterterm in the decoupling limit

\begin{equation}
\label{ctAAA}
    \mathcal{L}_{\rm{DL},\,ct}^{(3),\,\delta A^3}\sim \frac{\partial^4}{\Lambda_3^3\Lambda_4^4}\big[F_{\delta A}^2(\partial^2\phi)\big],
\end{equation}
which is safely suppressed by the strong coupling scales $\Lambda_3,\,\Lambda_4$ of the theory, such that the classical effective field theory structure remains unharmed.

\subsubsection{One-loop corrections to the cubic vector-tensor vertices}
The next class of diagrams contributes to the renormalization of the cubic mixed-sector interactions, involving either two Proca fields and one graviton or, the other way around, two gravitons and one Proca field as external legs. We first focus on corrections to three-point vertices of the type $\sim \delta A^2h$, where we omit the derivative structure for simplicity. Despite numerous contributing one-loop diagrams in the unitary-gauge formulation, there are only two distinct combinations that are allowed after taking the decoupling limit, which we depict in Fig. \ref{fig3ptAAh}.

\begin{figure}[H]
    \centering
    \includegraphics[width=0.7\linewidth]{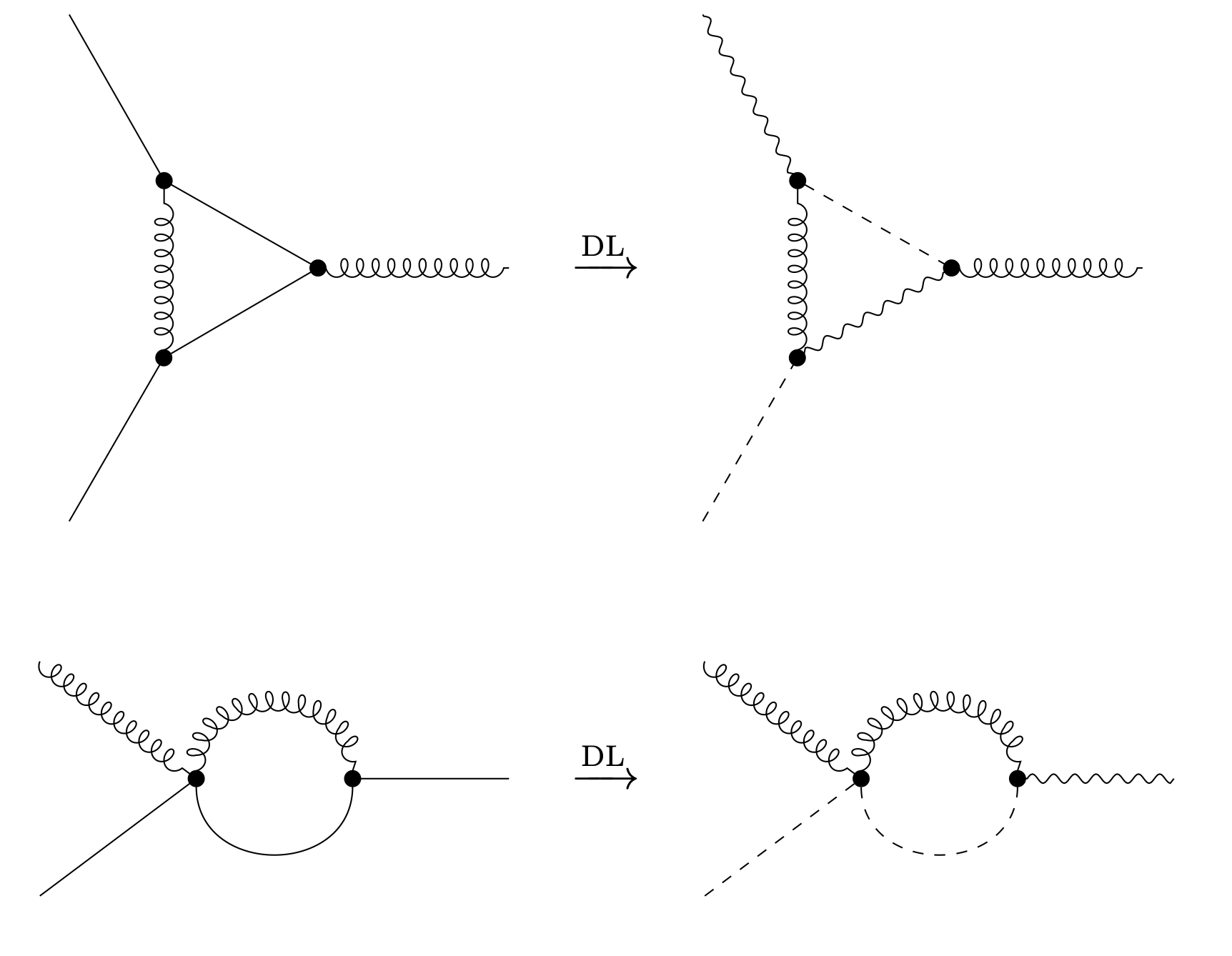}
    \caption{\small{The surviving one-loop 1PI diagrams in the decoupling limit that induce quantum corrections to the classical three-point vertex $\sim  \delta A^2h$. The left diagrams correspond to the one-loop contribution in unitary gauge, whereas the right ones result in the decoupling limit, with wavy lines denoting gauge-preserving vector legs or propagators $\sim F_{\delta A}$ and dashed lines scalar legs or propagators.}}
\label{fig3ptAAh}
\end{figure}
In unitary gauge, we anticipate the following one-loop structure through dimensional analysis and Lorentz invariance

\begin{equation}
    \begin{aligned}
        &\mathcal{M}_{(\mathcal{L}_{2}^{(3)})^3}\sim \frac{1}{ M_{\rm pl}^3}\left(m^4+m^2 p^2+p^4+\frac{p^6}{m^2}+\frac{p^8}{m^4}\right),\\[10pt]
        &\mathcal{M}_{\mathcal{L}_{2}^{(3)},\,\mathcal{L}_{2}^{(4)}}\sim \frac{1}{ M_{\rm pl}^3}\left(m^4+m^2 p^2+p^4+\frac{p^6}{m^2}\right),
    \end{aligned}
\end{equation}
where $p$ stands for external momenta and the expansion stops either at order $\sim 1/m^2$ or $\sim 1/m^4$, depending on the number of internal vector legs. Again, we have to pay special attention to the highest-order contributions $\sim p^6$ and $\sim p^8$ in external momenta, which could heavily detune the classical effective field theory structure at the quantum level.

Through explicit computation, we find the precise one-loop contribution of these two diagrams to the UV divergence of the off-shell reduced matrix element 

\begin{equation}
\begin{aligned}
\label{M3divAAh}
    \mathcal{M}^{\rm{div}}_{3,\,\delta A^2h}\supset\, & \frac{m^4}{960 \pi^2 \varepsilon M_{\rm pl}^3}\left[M_3^{(0)}(p_1,p_2)+\frac{1}{m^2}M_3^{(2)}(p_1,p_2)+\frac{1}{m^4}M_3^{(4)}(p_1,p_2)\right.\\[10pt]
    &\left.+\frac{1}{m^6}M_3^{(6)}(p_1,p_2)\right],
\end{aligned}
\end{equation}
where we adopted the same notation as above and simply give the leading contribution
\begin{equation}
    M_3^{(0)}=-480[E_1\epsilon_{23}+\epsilon_{12}E_3+\epsilon_{13} E_2+2(E\epsilon\epsilon_{123}+E\epsilon\epsilon_{213}+E\epsilon\epsilon_{312})],
\end{equation}
where $E_i^{\mu\nu}$ denotes the graviton polarization and we abbreviated $E_i\equiv E^\mu_{i,\,\mu}$ as well as $E\epsilon\epsilon_{ijk}\equiv E_i^{\mu\nu}\epsilon_{j,\,\mu}\epsilon_{k,\,\nu}$, with higher momentum terms $M_3^{(j)}$ relegated to the ancillary file. In Eq. \eqref{M3divAAh}, notice that the dangerous momentum powers $\sim p^6$ do not vanish explicitly. Therefore, the effective field theory structure is only protected if the associated counterterms at high energies feature the vector field in gauge-preserving form $F_{\delta A}$.\footnote{Due to the enormous complexity of the expression in \eqref{M3divAAh}, we refrain from showing this statement explicitly in the unitary gauge calculation.} Indeed, performing a power counting analysis in dimensional regularization of the corresponding one-loop graphs depicted on the right-hand side of Fig. \ref{fig3ptAAh} yields the following counterterms in the decoupling limit

\begin{equation}
\label{ctAAh}
    \mathcal{L}_{\rm{DL},\,ct}^{(3),\,\delta A^2 h}\sim \frac{\partial^4}{\Lambda_4^6}\big[F_{\delta A}(\partial h)(\partial\phi)\big]+\frac{\partial^4}{\Lambda_4^6}\big[F_{\delta A}\,h\,(\partial^2\phi)\big],
\end{equation}
where one vector field manifestly enters in gauge-invariant form, which tames the bad high-energy behaviour of the possible counterterms expected from unitary-gauge power counting. Furthermore, both quantum-induced interactions are safely under control owing to the strong coupling scale $\Lambda_4$, such that the mass of the associated ghosts does not come close to the effective field theory cutoff. 

The second class of cubic vector-tensor interactions describes vertices with two gravitons and one Proca field that are of the schematic form $\sim h^2\delta A$, where we again suppressed the precise derivative structure for simplicity. For this type of interaction, there are two remnant one-loop diagrams in the decoupling limit given in Fig. \ref{fig3pthhA}.

\begin{figure}[H]
    \centering
    \includegraphics[width=0.7\linewidth]{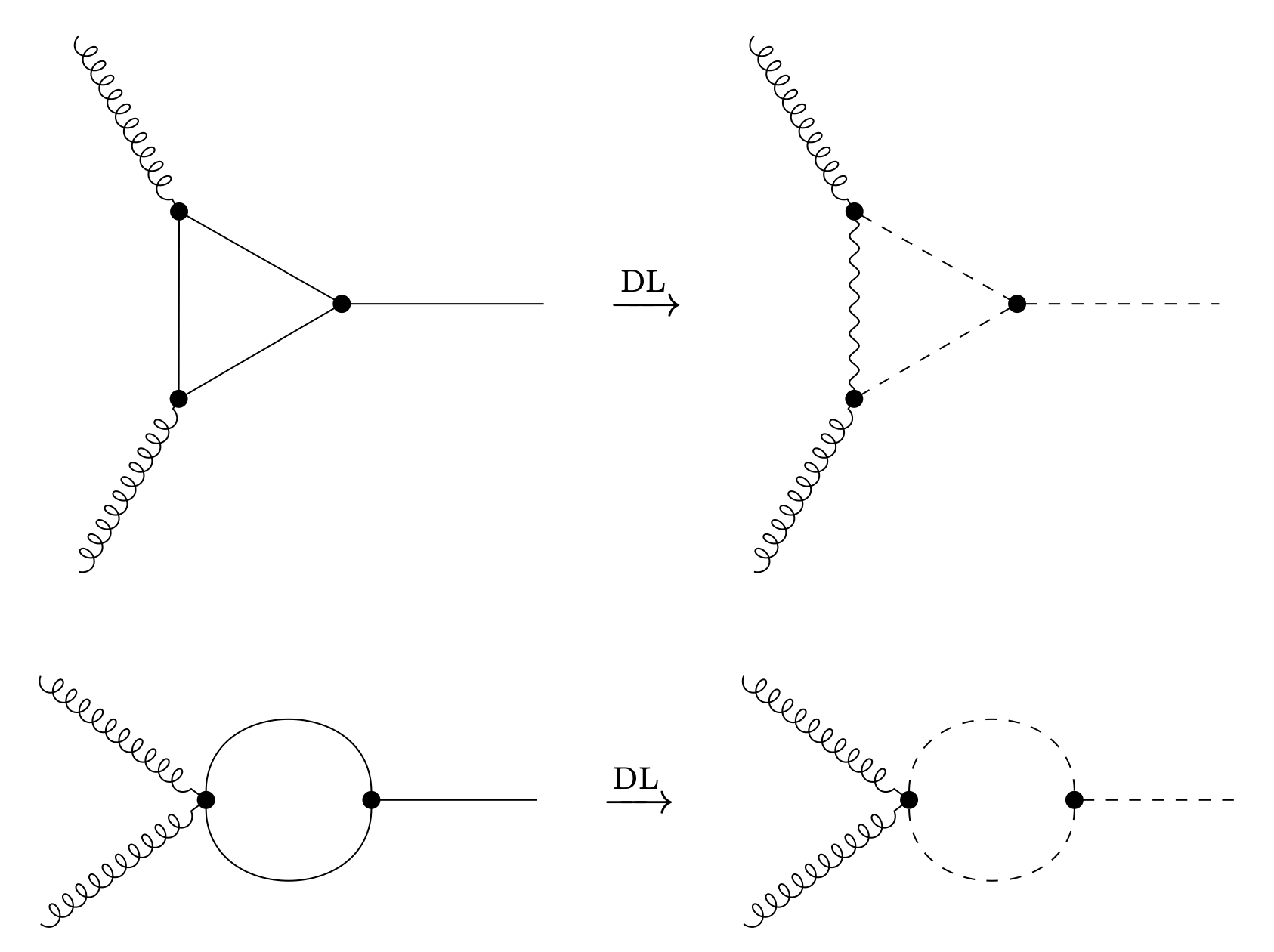}
    \caption{\small{The surviving one-loop 1PI diagrams in the decoupling limit yielding quantum corrections to the three-point vertex $\sim h^2 \delta A$. The left diagrams are in unitary gauge, whereas the right ones result in the decoupling limit, with wavy lines denoting gauge-preserving vector legs $\sim F_{\delta A}$ and dashed lines scalar legs or propagators.}}
\label{fig3pthhA}
\end{figure}
Based on power counting arguments in unitary gauge within dimensional regularization and Lorentz invariance, we expect the following results for these one-loop graphs

\begin{equation}
    \begin{aligned}
        &\mathcal{M}_{(\mathcal{L}_{2}^{(3)})^2,\,\mathcal{L}_{3}^{(3)}}\sim \frac{m^2}{ M_{\rm pl}^2\Lambda_2^2}\left(m^2 p+p^3+\frac{p^5}{m^2}+\frac{p^7}{m^4}+\frac{p^9}{m^6}\right),\\[10pt]
        &\mathcal{M}_{\mathcal{L}_{2}^{(3)},\,\mathcal{L}_{2}^{(4)}}\sim \frac{m^2}{ M_{\rm pl}^2\Lambda_2^2}\left(m^2 p+p^3+\frac{p^5}{m^2}+\frac{p^7}{m^4}\right),
    \end{aligned}
\end{equation}
where the highest power in external momenta $\sim p^9$ can lead to severe quantum detuning. Through explicit calculation, we obtain the following UV-divergent contribution to the reduced matrix element with off-shell legs

\begin{equation}
\begin{aligned}
\label{M3hhAdiv}
    \mathcal{M}^{\rm{div}}_{3,\,h^2\delta A}\supset\, & \frac{-i\,c_3\,m^4}{480 \pi^2 \varepsilon M_{\rm pl}^2\Lambda_2^2}\left[M_3^{(1)}(p_1,p_2)+\frac{1}{m^2}M_3^{(3)}(p_1,p_2)+\frac{1}{m^4}M_3^{(5)}(p_1,p_2)\right.\\[10pt]
    &\left.+\frac{1}{m^6}M_3^{(7)}(p_1,p_2)\right],
\end{aligned}
\end{equation}
with the detailed structure outsourced to the ancillary file, we observe that the dangerous $\sim p^9$ terms are manifestly absent, indicating once again that non-trivial cancellations in unitary gauge seem to magically rescue the effective field theory. Yet, these cancellations can be understood from a  decoupling limit perspective: dimensional analysis together with Lorentz invariance applied to the resulting right-hand diagrams in Fig. \ref{fig3pthhA} within dimensional regularization induces the following one-loop corrections at the Lagrangian level

\begin{equation}
\label{LhhADL}
    \mathcal{L}_{\rm{DL},\,ct}^{(3),\,h^2\delta A}\sim\, 
    \frac{\partial^6}{\Lambda_3^3\Lambda_4^4}\big[h\,(\partial h)(\partial\phi)\big]+
    \frac{\partial^4}{\Lambda_3^3\Lambda_4^4}\big[(\partial h)^2(\partial^2\phi)\big],
\end{equation}
which are safely suppressed by the strong coupling scales $\Lambda_3,\,\Lambda_4$ of the theory in the decoupling limit. Furthermore, Eq. \eqref{LhhADL} explains why there was no $\sim p^9$ in the reduced matrix element \eqref{M3hhAdiv}, simply because terms of this order cannot be formed within the decoupling limit. Overall, we conclude that the classical effective field theory structure in the decoupling limit remains unharmed from one-loop corrections to tree-level diagrams in the mixed Proca-graviton sector.

\subsubsection{One-loop corrections to the cubic graviton vertex}
The last group of one-loop diagrams features only external graviton legs, which gives quantum contributions to the classical three-point interactions $\sim (\partial h)^2h,\,$ in the pure graviton sector. Despite multiple one-loop graphs partaking in the renormalization the three-point graviton vertex in unitary gauge, a quick inspection of the surviving vertices \eqref{LagrDL} shows that there are no surviving diagrams once the decoupling limit is taken. Thus, the cubic pure-graviton interaction does suffer from quantum detuning in the decoupling limit, such that its classical structure remains manifestly protected. 

To summarize, we have witnessed quantum corrections at one loop to be parametrically controlled by the effective field theory cutoffs of the theory in the decoupling limit. Despite worrisome expectations based on dimensional analysis for the full theory in unitary gauge, actual computations of one-loop corrections revealed that most dangerous high-momentum terms are absent, although they would have technically been allowed. These non-trivial cancellations can readily be understood from the decoupling limit, where terms of this order cannot be constructed to begin with. In addition, the decoupling limit analysis shows explicitly that all high-momentum contributions that persevere at high energies originate from counterterms involving the Proca field in gauge-invariant form $F_{\delta A}$, whose improved high-energy behaviour does not destabilize the effective field theory structure. Moreover, all quantum-induced operators so far have come with more derivatives than their corresponding classical pendants, which implies that the classical structure of Generalized Proca theory generically remains protected from renormalization.




\subsection{A glimpse at four-point functions}
\label{sec7}
In this section, we turn our attention to one-loop corrections to the classical four-point functions arising in the minimal model \eqref{minmodel}. Due to the increasing number of contributing diagrams and the complexity of the loop integrals, which comes with high computational costs at quartic order, we restrict our analysis to a general power counting argument in the decoupling limit in order to infer the allowed quantum counterterms for a number of selected diagrams. The one-loop 1PI graphs displayed in Fig. \ref{fig4pthhss} can consistently be constructed from the decoupling limit vertices \eqref{LagrDL} and yield corrections to the surviving quartic scalar-tensor interaction in \eqref{LagrDL} featuring two scalar fields and two gravitons.

\begin{figure}[H]
    \centering
    \includegraphics[width=0.85\linewidth]{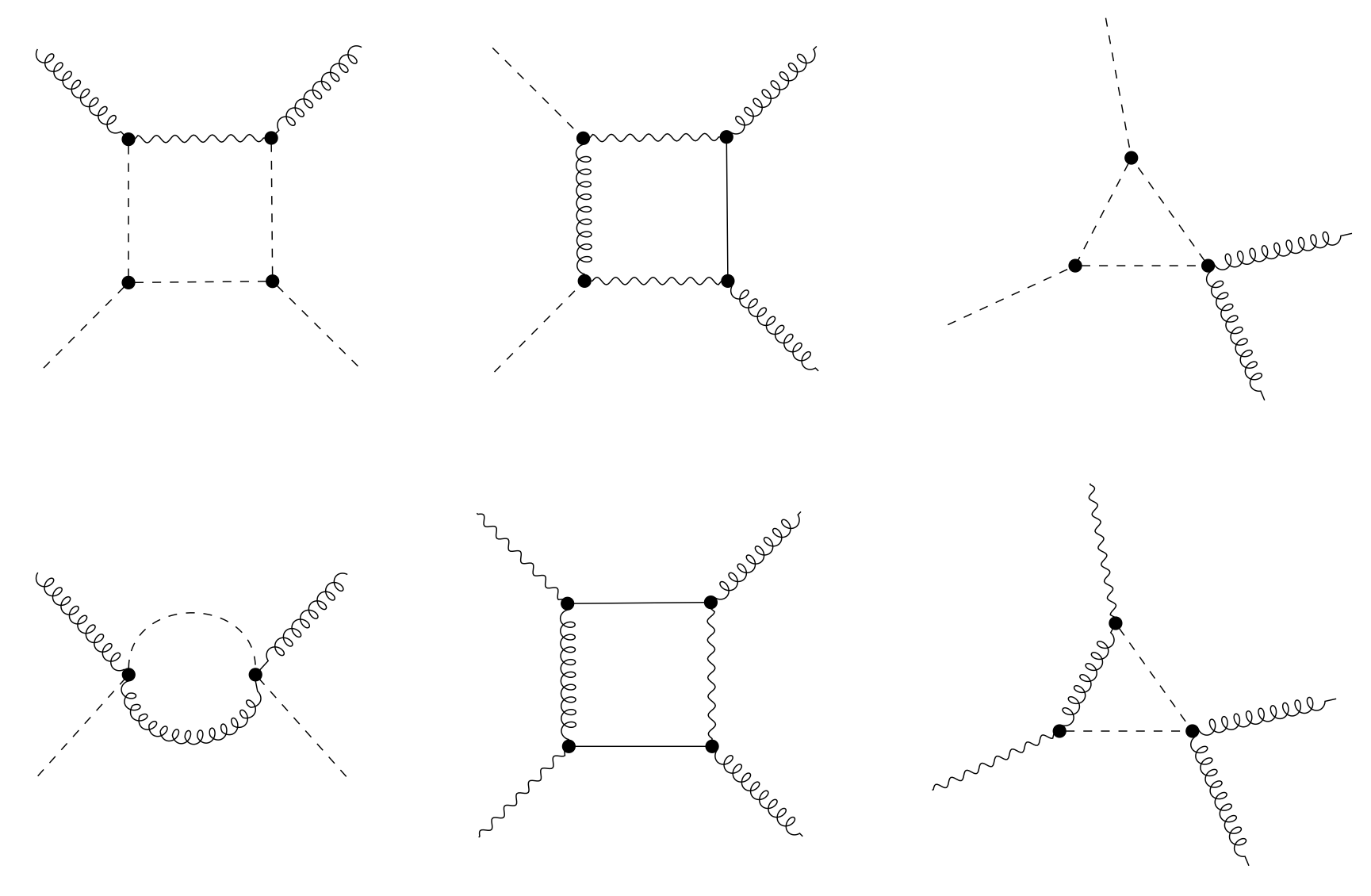}
    \caption{\small{All mixed-sector one-loop 1PI diagrams contributing to the renormalization of the remaining quartic scalar-tensor vertex in the decoupling limit. These diagrams are purely constructed from the $\mathcal{L}^{(3)}_{\rm{DL}}$ and $\mathcal{L}^{(4)}_{\rm{DL}}$ vertices, where the box diagram can either have gauge-preserving transverse vector polarizations $\sim F_{\delta A}$, denoted by wavy lines, or scalar fields as external legs, corresponding to dashed lines. }}
\label{fig4pthhss}
\end{figure}

Invoking dimensional analysis and Lorentz invariance, we find the corresponding quantum-induced counterterms

\begin{equation}
\label{cthhAA}
\begin{aligned}
    \mathcal{L}_{\rm DL,\,ct}^{(4),\,h^2\delta A^2}\sim&\frac{\partial^4}{\Lambda_4^4\Lambda_3^6}\big[(\partial^2\phi)^2(\partial h)^2\big]+\frac{\partial^6}{\Lambda_4^4\Lambda_3^6}\big[(\partial\phi)^2(\partial h)^2+h\,(\partial^2\phi)(\partial\phi) (\partial h)\big]\\[5pt]
    &+\frac{\partial^8}{\Lambda_4^4\Lambda_3^6}\big[h^2(\partial\phi)^2\big]+\frac{\partial^6}{\Lambda_4^8}\big[ h^2F_{\delta A}^2\big]+ \frac{\partial^4}{\Lambda_4^8}\big[(\partial h)^2F_{\delta A}^2\big]\\[5pt]
    &+\frac{\partial^4}{\Lambda_4^8}\big[(\partial h)^2(\partial\phi)^2+h(\partial^2\phi)(\partial h)(\partial\phi)+h^2(\partial^2\phi)^2\big],
\end{aligned}
\end{equation}
which are all safely suppressed by the strong coupling scales of the Proca effective field theory, such that they do not spoil the crucial hierarchy between classical and quantum terms. 

Furthermore, the classical vertices remaining in the decoupling limit allow for the construction of new operators at the quantum level that do not renormalize any classical interaction. For example, three insertions of $\mathcal{L}_{2,\,\rm DL}^{(3)}$ combined with $\mathcal{L}_{3,\,\rm DL}^{(3)}$ enable one-loop diagrams that induce quantum terms describing interactions among one graviton and three scalars, as shown in Fig. \ref{fig4pthsss}.

\begin{figure}[H]
    \centering
    \includegraphics[width=0.3\linewidth]{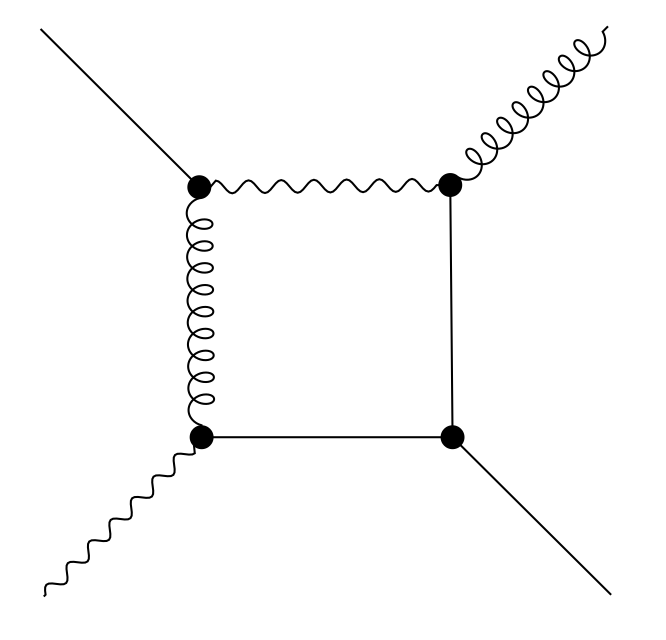}
    \caption{\small{Mixed-sector one-loop 1PI diagram creating a new scalar-tensor interaction at the quantum level in the decoupling limit, constructed from $\mathcal{L}^{(3)}_{2,\,\rm{DL}}$ and $\mathcal{L}^{(3)}_{3,\,\rm{DL}}$ vertices. Wavy lines denote gauge-preserving transverse vector polarizations $\sim F_{\delta A}$, dashed lines correspond to scalar legs. }}
\label{fig4pthsss}
\end{figure}

The corresponding quantum-induced operators found in power counting together with Lorentz invariance reads

\begin{equation}
\label{cthAAA}
    \mathcal{L}_{\rm DL,\,ct }^{(4),\,h\,\delta A^3}\sim \frac{\partial^4}{\Lambda_4^6\Lambda_3^3}\big[(\partial^2\phi)(\partial\phi)(\partial h) F_{\delta A}\big]+\frac{\partial^6}{\Lambda_4^6\Lambda_3^3}\big[(\partial\phi)^2h\,F_{\delta A}\big],
\end{equation}
which are healthily suppressed by the strong coupling scales $\Lambda_3,\,\Lambda_4$ of Generalized Proca theory.

Clearly, the diagram in Fig. \ref{fig4pthsss} is not the only one capable of inducing genuinely new quartic interactions at the quantum level. The decoupling limit Lagrangians \eqref{LagrDL} also generate pure Proca and pure graviton four-point vertices from a number of bubble, triangle and box diagrams. Instead of repeating the same power counting analysis for these 1PI graphs again to show that these scalar-vector-tensor interactions are safely controlled by the effective field theory cutoffs $\Lambda_3,\,\Lambda_4$, we will establish a non-renormalization theorem and show its robustness for the full Generalized Proca effective field theory.




\section{Non-Renormalization Theorem}
\label{NR-Theorem}
We are now in a position to summarize the results of the preceding sections in the form of a non-renormalization theorem governing Generalized Proca effective field theories coupled to gravity.
 \begin{theorem}[Non-Renormalization of Generalized Proca]\label{lem:NR}
 Consider a Generalized Proca effective field theory coupled to gravity and expanded about Minkowski spacetime, with Proca mass 
$m$, Planck mass $M_{\rm pl}$, and a hierarchy of strong-coupling scales defining an effective field theory cutoff 
$\Lambda_{\rm EFT}$. Assume that the classical interaction terms are chosen so as to propagate only the three physical degrees of freedom of a massive vector field and to be free of Ostrogradski instabilities at the classical level.
Then, within the regime of validity of the effective field theory and for energies $E\ll\Lambda_{\rm EFT}$ quantum corrections do not renormalize the classical ghost-free Proca interactions. More precisely:
\begin{itemize}
\item All local operators generated radiatively by graviton-mediated and matter loops necessarily contain additional derivatives per field relative to the classical interaction terms and are suppressed by the strong-coupling scales of the theory.
\item No operators with derivatives per field as those defining the classical Generalized Proca action are generated with coefficients unsuppressed by the effective field theory cutoff.
\item Any higher-derivative operators induced by quantum corrections give rise only to ghost-like excitations with masses parametrically above $\Lambda_{\rm EFT}$ and therefore do not affect the low-energy spectrum or dynamics.
\end{itemize}
Consequently, the operator basis defining ghost-free Generalized Proca theories is closed under radiative corrections within the effective field theory expansion, and the theories are radiatively stable to all orders in perturbation theory when coupled to gravity.
\end{theorem}
 \begin{corollary}[Invariance of Low-Energy Degrees of Freedom]\label{coro:NR}
The low-energy dynamics and number of propagating degrees of freedom of Generalized Proca theories remain invariant under quantum corrections below the effective field theory cutoff.
\end{corollary}
\begin{proof}
The proof proceeds in three steps and relies on the effective field theory hierarchy of scales, the existence of a well-defined decoupling limit, and the classification of radiatively generated local operators.
\paragraph{$(i)$ Decoupling-limit control of counterterms:}
As established in Lemma \ref{lem:DL}, the decoupling limit defines a controlled scaling regime in which all operators contributing to physical amplitudes at energies 
$E\ll\Lambda_{\rm EFT}$ must survive with finite coefficients. Any local operator that vanishes or becomes parametrically suppressed in this limit cannot be generated radiatively with an unsuppressed coefficient below the cutoff. The decoupling limit therefore provides a necessary condition for the appearance of admissible counterterms in the effective action. Its robustness for the entire class will be further confirmed in \ref{robustness}.
\paragraph{$(ii)$ Operator classification at the quantum level:}
The explicit one-loop computations of two- and three-point functions performed in Section \ref{quantumcorrections} exhaust the set of radiatively generated local operators compatible with Lorentz invariance, locality, and the symmetries of the theory. In all cases, the resulting operators contain additional derivatives per field relative to the classical ghost-free interactions and are suppressed by the strong-coupling scales identified in the decoupling limit. In particular, no operators with derivatives per field as those defining the classical Generalized Proca action are generated. This establishes that the classical operator basis is closed under radiative corrections at one loop within the effective field theory expansion.
\paragraph{$(ii)$ Extension to all loop orders:}
The closure of the operator basis and the derivative suppression identified at one loop extend to all orders in perturbation theory. Indeed, higher-loop diagrams can only generate local counterterms built from the same operator structures allowed by the decoupling limit and are necessarily further suppressed by additional powers of the strong-coupling scales. Since no lower-derivative operators survive in the decoupling limit, none can be generated radiatively with coefficients unsuppressed by 
$\Lambda_{\rm EFT}$ at any loop order (see section \ref{robustness}). 

Combining these results, we conclude that quantum corrections do not renormalize the classical ghost-free Proca interactions below the effective field theory cutoff. Any higher-derivative operators induced radiatively correspond to degrees of freedom with masses parametrically above $\Lambda_{\rm EFT}$ and therefore do not affect the low-energy spectrum or dynamics. This establishes the non-renormalization and radiative stability of Generalized Proca theories coupled to gravity within the effective field theory framework.
\end{proof}

 \begin{remark}[Scope of gravitational coupling]\label{rem:gravity}
The non-renormalization theorem is established for Generalized Proca effective field theories coupled to gravity in the weak-field regime, i.e. expanded around a specific spacetime and treated perturbatively in the gravitational fluctuations. The proof relies only on the derivative structure of the interactions and on effective field theory power counting, and does not assume any special properties of non-linear or strongly curved gravitational backgrounds. For such backgrounds Vainshtein screening at quantum level will be relevant.
 \end{remark}
  \begin{remark}[Stability under field redefinitions]\label{rem:redefinition}
The conclusions of Theorem \ref{NR-Theorem} are invariant under local field redefinitions of the Proca and metric fields that preserve the effective field theory expansion. Such redefinitions may reshuffle operators within the effective action but cannot generate lower-derivative interactions or modify the operator classification below the cutoff. Consequently, the non-renormalization property is a statement about the equivalence class of ghost-free Generalized Proca theories rather than a specific choice of Lagrangian representation.
\end{remark}
  \begin{remark}[Matter couplings]\label{rem:matter}
The argument extends straightforwardly to the inclusion of minimally coupled matter fields whose interactions respect locality and the effective field theory hierarchy of scales. Matter loops can generate additional higher-derivative operators but are subject to the same decoupling-limit constraints and therefore do not alter the non-renormalization of the classical Proca interactions below the effective field theory cutoff.
 \end{remark}
 \begin{remark}[Role of the effective field theory cutoff]\label{rem:cutoff}
The non-renormalization theorem is inherently an effective field theory statement and applies only at energies parametrically below the cutoff 
$\Lambda_{\rm EFT}$. At energies approaching the cutoff, higher-derivative operators become important and the effective description necessarily breaks down. The theorem therefore does not address the ultraviolet completion of Generalized Proca theories, but rather constrains their infrared quantum consistency.
\end{remark}
 \begin{remark}[Relation to flat-space results]\label{rem:flatresults}
In the limit of vanishing gravitational interactions, Theorem \ref{NR-Theorem} reduces to the known non-renormalization properties of Generalized Proca and Galileon theories in flat spacetime. The present analysis shows that these structural properties persist under the inclusion of weak gravitational fields, despite the absence of an exact gauge symmetry protecting the massive vector interactions.
\end{remark}




\section{Structural Origin and Robustness of the Non-Renormalization Theorem}
\label{robustness}
In this section we provide a structural interpretation of the non-renormalization theorem established above and discuss the robustness of the result beyond the explicit perturbative computations.
While the minimal model chosen in \eqref{minmodel} allows for tractable computations of one-loop corrections, in this section, we include all the Lagrangians $\mathcal{L}_{i\leq 6}$ from \eqref{gpLagr} with more generic coupling functions. Since we are interested in understanding how quantum corrections involving the graviton affect the non-renormalization theorem for the Proca effective field theory on flat space \cite{Heisenberg:2020jtr}, we require the coupling functions $G_i$ in the covariant theory \eqref{gpLagr} to be structured such that the corresponding Generalized Proca interactions on flat space, meaning the vector Galileons, are recovered as part of the perturbative expansion around the background configuration \eqref{bgconfig}. In general, we assume that the coupling functions $G_i$ are smooth and admit the following Taylor expansions
\begin{equation}
\label{Gexp}
    \begin{aligned}
        &G_2(X,Y,F)=\bar{G}_2(0,0,0)+\bar{G}_{2,X}(0,0,0)\delta A^2+\bar{G}_{2,F}(0,0,0)F_{\delta A}^2+\mathcal{O}(\delta A^4),\\[10pt]
        &G_i(X)=\bar{G}_i(0)+\bar{G}_{i,X}(0)\delta A^2+\bar{G}_{i,XX}(0)\delta A^4+\mathcal{O}(\delta A^6)\qquad\quad\qquad\, \rm{for}\quad i\geq 3,\\[10pt]
        &G_{i,X}(X)=\bar{G}_{i,X}(0)+\bar{G}_{i,XX}(0)\delta A^2+\bar{G}_{i,XXX}(0)\delta A^4+\mathcal{O}(\delta A^6)\qquad \rm{for}\quad i\geq 4,
    \end{aligned}
\end{equation}
where we have expanded perturbatively in the vector field around the background configuration \eqref{bgconfig}, which further imposes $\bar{G}_2=0$. The couplings $\bar{G}_i$, $\bar{G}_{i,X}$ evaluated on the background become constant parameters that encode the effective field theory coefficients and need to be chosen in terms of the classical scales $m$, $\Lambda_2$ and $M_{\rm pl}$ such that the resulting Lagrangians have the correct mass dimension. Note that the coefficients $\bar{G}_{2,F}(0,0,0)$ and $\bar{G}_4(0)$ contain the normalization of the vector and graviton kinetic terms, respectively. Regarding the functions $G_2,\,G_3$ and $g_5$, their effective field theory expansion parallels the one on flat space \cite{Heisenberg:2020jtr}

\begin{equation}
\label{Giexp}
    \begin{aligned}
        &G_2(X,Y,F)=\Lambda_2^4\left(c_{2,1}^F\frac{F_{\delta A}^2}{\Lambda_2^4}+c_{2,1}^X\frac{m^2\delta A^2}{\Lambda_2^4}+c_{2,2}^X\frac{m^4\delta A^4}{\Lambda_2^8}+c_{2,2}^F\frac{F_{\delta A}^4}{\Lambda_2^8}+...\right)\\[10pt]
        &G_3(X)=\Lambda_2^2\left(c_{3,0}+c_{3,1}\frac{m^2\delta A^2}{\Lambda_2^4}+c_{3,2}\frac{m^4\delta A^4}{\Lambda_2^8}+...\right),\\[5pt]
        &\qquad\quad = \frac{m^2\delta A^2}{\Lambda_2^2}\left(c_{3,1}+c_{3,2}\frac{m^2\delta A^2}{\Lambda_2^4}+...\right),\\[10pt]
        &g_5(X)=\frac{1}{\Lambda_2^2}\left(\tilde{c}_{5,0}+\tilde{c}_{5,1}\frac{m^2\delta A^2}{\Lambda_2^4}+\tilde{c}_{5,2}\frac{m^4\delta A^4}{\Lambda_2^8}+...\right),
    \end{aligned}
\end{equation}
where the $c_{i,j},\,\tilde{c}_{i,j}$ denote the dimensionless effective field theory coefficients that precede powers of $\delta A^2$ in the dimensionless combination $\frac{m^2\delta A^2}{\Lambda_2^4}$, as required by the effective field theory rulebook. Conventionally, we set $c_{2,1}^X=1$ such that loop contributions give corrections to the squared mass $m^2$ directly. Furthermore, we suppress the explicit expansion of $G_2$ in $Y$, as we can schematically absorb these terms through suitable multiplications of $\frac{m^2\delta A^2}{\Lambda_2^4}$ with $\frac{F_{\delta A}^2}{\Lambda_2^4}$, as further discussed in appendix \ref{app4}. The second line in the equation for $G_3$ follows from the fact that terms proportional to $\bar{G}_3(0)$ amount to total derivatives and hence drop out of the analysis, such that we can set $\bar{G}_3(0)=c_{3,0}=0$ without loss of generality. 

Since the vector self-interactions in the Lagrangians $\mathcal{L}_{4,5,6}$ \eqref{gpLagr} are encoded in the terms proportional to the derivatives of the coupling functions $G_{i,X}$ for $i=4,5,6$, we require the perturbative expansion of $G_{i,X}$ in \eqref{Gexp} to reproduce the Generalized Proca effective field theory on flat space, which looks as follows
\begin{equation}
\label{GiXexp}
    \begin{aligned}
        &G_{4,X}(X)=c_{4,1}+c_{4,2}\frac{m^2\delta A^2}{\Lambda_2^4}+c_{4,3}\frac{m^4\delta A^4}{\Lambda_2^8}+...,\\[10pt]
        &G_{5,X}(X)=\frac{1}{\Lambda_2^2}\left(c_{5,1}+c_{5,2}\frac{m^2\delta A^2}{\Lambda_2^4}+c_{5,3}\frac{m^4\delta A^4}{\Lambda_2^8}+...\right),\\[5pt]
        &\qquad\qquad = \frac{1}{\Lambda_2^2}\left(c_{5,2}\frac{m^2\delta A^2}{\Lambda_2^4}+c_{5,3}\frac{m^4\delta A^4}{\Lambda_2^8}+...\right),\\[10pt]
        &G_{6,X}(X)=\frac{1}{\Lambda_2^4}\left(c_{6,1}+c_{6,2}\frac{m^2\delta A^2}{\Lambda_2^4}+c_{6,3}\frac{m^4\delta A^4}{\Lambda_2^8}+...\right),
    \end{aligned}
\end{equation}
where, again, the prefactors involving powers of $\sim 1/\Lambda_2$ take care of the Lagrangian dimension. Here, terms proportional to $\bar{G}_{5,X}(0)\sim c_{5,1}$ end up being total derivatives and can safely be neglected, so we set $c_{5,1}=0$ without loss of generality \cite{Heisenberg:2014rta, Heisenberg:2020jtr}. To obtain the full coupling functions, we still need to integrate \eqref{GiXexp} with respect to $X$. As pointed out in \cite{Heisenberg:2014rta,DeFelice:2016yws}, vector Galileons on curved spacetimes start at quadratic order $G_4(X)\sim X^2$ without linear dependence on $X$, so we additionally set $c_{4,1}=0$ to unambiguously recover the flat space quartic Generalized Proca interaction. Overall, the full coupling functions take the schematic form
\begin{equation}
\begin{aligned}
\label{vecgal}
    &G_2(X,Y,F)=(F-m^2 X)\left(\frac{m^2 X}{\Lambda_2^4}\right)^{a_1}\left(\frac{F}{\Lambda_2^2}\right)^{a_2},\\[5pt]
    &G_3(X)=\frac{m^2 X}{\Lambda_2^2}\left(c_{3,1}+c_{3,2}\frac{m^2 X}{\Lambda_2^4}+...\right),\\[8pt]
    &G_4(X)=\frac{M_\text{pl}^2}{2}+\frac{m^2 X^2}{\Lambda_2^4}\left(c_{4,2}+c_{4,3}\frac{m^2 X}{\Lambda_2^4}+...\right),\\[8pt]
    &G_5(X)=c_{5,0}+\frac{m^2 X^2}{\Lambda_2^6}\left(c_{5,2}+c_{5,3}\frac{m^2 X}{\Lambda_2^4}+...\right),\\[8pt]
    &g_5(X)=\frac{1}{\Lambda_2^2}\left(\tilde{c}_{5,0}+\tilde{c}_{5,1}\frac{m^2 X}{\Lambda_2^4}+...\right),\\[8pt]
    &G_6(X)=\frac{c_{6,0}}{\Lambda_2^2}+\frac{X}{\Lambda_2^4}\left(c_{6,1}+c_{6,2}\frac{m^2 X}{\Lambda_2^4}+...\right),
\end{aligned}
\end{equation}
where we have set $c_{4,0}= M_{\rm pl}^2/2$ and $c_{2,1}^F=1$ to canonically normalize the kinetic terms of the graviton and the Proca field. Although all Lorentz contractions are suppressed, Eq. \eqref{vecgal} should be understood in the sense that the theory assumes three propagating degrees of freedom. The dimensionless factors raised to powers $a_1,\, a_2\geq 0$ are implicitly summed over as in \eqref{Giexp} and include the flat-space effective field theory structure on the one hand but also lead to non-trivial interactions with the graviton on the other, particularly through expansions of the Proca kinetic term. The resulting Lagrangians at third and fourth order in perturbations are schematically given in Appendix \ref{app4}. 
To go beyond the vector Galileon model, consider the following very generic effective field theory expansion for the coupling functions
\begin{equation}
\begin{aligned}
\label{genmodel}
    &G_2(X,Y,F)=(F-m^2 X)\left(\frac{m^2 X}{\Lambda_2^4}\right)^{a_1}\left(\frac{F}{\Lambda_2^2}\right)^{a_2},\\[5pt]
    &G_3(X)=\frac{m^2 X}{\Lambda_2^2}\left(c_{3,1}+c_{3,2}\frac{m^2 X}{\Lambda_2^4}+...\right),\\[8pt]
    &G_4(X)=\frac{M_\text{pl}^2}{2}+\Lambda_2^2\left(c_{4,1}\frac{m^2 X}{\Lambda_2^4}+c_{4,2}\frac{m^4 X^2}{\Lambda_2^8}+...\right),\\[8pt]
    &G_5(X)=c_{5,0}+c_{5,2}\frac{m^4 X^2}{\Lambda_2^8}+c_{5,1}\frac{m^6 X^3}{\Lambda_2^{12}}+...,\\[8pt]
    &g_5(X)=\frac{1}{\Lambda_2^2}\left(\tilde{c}_{5,0}+\tilde{c}_{5,1}\frac{m^2 X}{\Lambda_2^4}+...\right),\\[8pt]
    &G_6(X)=\frac{1}{\Lambda_2^2}\left(c_{6,0}+c_{6,1}\frac{m^2 X}{\Lambda_2^4}+...\right),
\end{aligned}
\end{equation}
where we directly set $c_{3,0}=c_{5,1}=0$, as the corresponding terms in the Lagrangians amount to total derivatives and can hence be disregarded. The classical scales take care of the right Lagrangian dimension and the expansion again consists of the dimensionless building blocks $\frac{m^2X}{\Lambda_2^4}$ and $\frac{F}{\Lambda_2^2}$. This second model describes the most straightforward effective field theory expansion one can write down for Generalized Proca theory.
However, since factors of $X$ accompanying vector self-interactions are suppressed by higher powers of $m^2/\Lambda_2^4$ compared to the vector Galileon model \eqref{vecgal}, there are no surviving interaction terms beyond the ones stemming from the expansion of $\mathcal{L}_2$ in the decoupling limit, rendering the theory \eqref{genmodel} a subclass of the vector Galileon model at high energies. Therefore, we proceed by establishing radiative stability for the general vector Galileon model \eqref{vecgal}. Since this model is adjusted to explicitly contain the flat-space interactions, a large part of the analysis will parallel the derivation on flat space. Therefore, it is useful to briefly remind ourselves how the non-renormalization theorem was established on flat space, which we elaborate on in the next section.

\subsection{Reviewing the non-renormalization on flat space}
On flat space, any Generalized Proca effective field theory can be written as a schematic expansion in four essential building blocks \cite{Heisenberg:2020jtr}

\begin{equation}
\label{gpflat}
    \mathcal{L}_{\rm{flat}}\sim(F_{\delta A}^2+m^2\delta A^2)\left(\frac{m^2\delta A^2}{\Lambda_2^4}\right)^{a_1}\left(\frac{F_{\delta A}^2}{\Lambda_2^4}\right)^{a_2}\left(\frac{\partial\,\delta A}{\Lambda_2^2}\right)^{a_3},\qquad a_1,a_2,a_3\geq 0
\end{equation}
where all quantities are Lorentz-contracted such that the theory remains ghost-free and propagates three degrees of freedom. The kinetic and mass term upfront ensure the correct Lagrangian dimension and are supplemented by dimensionless factors $m^2\delta A^2/\Lambda_2^4$ and $F_{\delta A}^2/\Lambda_2^4$ stemming from expansions of the coupling functions $G_{i\geq 2}$ and $g_5$. The last term is restricted to $0\leq a_3\leq 3$ and enters accordingly in the higher Lagrangians $\mathcal{L}_{4,5,6}$. The block representation \eqref{gpflat} is particularly useful to analyse the theory's behaviour in the decoupling limit, as each classical operator can be studied individually, resulting in

\begin{equation}
\label{LDLflat}
    \mathcal{L}_{\rm flat,\,DL}\sim (F_{\delta A}^2+(\partial\phi)^2)\left(\frac{\partial^2\phi}{\Lambda_3^3}\right)^{a_3},
\end{equation}
with terms proportional to $F_{\delta A}^2$ restricted to $a_3\leq 2$. Constructing loop diagrams from the decoupling limit \eqref{LDLflat} implies that each vertex contributes a factor of at least $\sim 1/\Lambda_3^3$. For one-loop 1PI diagrams treated in dimensional regularization, all induced counterterms consist of the two essential building blocks $\frac{\partial F_{\delta A}}{\Lambda_3^3}$ and $\frac{\partial^2\phi}{\Lambda_3^3}$, such that any quantum operator admits the following form

\begin{equation}
\label{ctDL}
    \mathcal{L}_{\rm flat,\,DL}^{\rm ct}\sim \partial^4\left(\frac{\partial F_{\delta A}}{\Lambda_3^3}\right)^{2b_2}\left(\frac{\partial^2\phi}{\Lambda_3^3}\right)^{b_3}
\end{equation}
where $2b_2+b_3=N\geq 2$ with $N$ the number of external legs and $b_{2,3}\geq 0$ denoting positive integers. Overall, the flat-space Proca EFT arranges itself as expansion in three distinct parameters

\begin{equation}
\label{flatEFTexp}
    \mathcal{L}_{\rm flat,\,DL}\sim (F_{\delta A}^2+(\partial\phi)^2)\,\alpha_{\rm cl}^{a_3}+(F_{\delta A}^2+(\partial\phi)^2)\,\alpha_{\rm q}^{2+n}\,\alpha_{\rm \tilde{q}}^l\, \alpha_{\rm cl}^{m},\quad a_3,m,n,l\geq 0,
\end{equation}
where the classical and quantum parameters are
\begin{equation}
\label{alphas}
    \alpha_{\rm cl}\equiv \frac{\partial^2\phi}{\Lambda_3^3},\quad \alpha_{\rm q}\equiv\frac{\partial^2}{\Lambda_3^2},\quad \alpha_{\rm \tilde{q}}\equiv\frac{F_{\delta A}^2}{\Lambda_3^4}.
\end{equation}

While $\alpha_{\rm cl}$ and $\alpha_{\rm q}$ are known from Galileon theories, where they control the classical, non-linear interactions and quantum-induced operators, respectively, the quantum parameter $\alpha_{\rm \tilde{q}}$ is genuinely new in vector theories. Importantly, Eq. \eqref{flatEFTexp} shows that the induced counterterms at the quantum level always come with additional derivatives compared to the classical interactions, implying that the classical structure of Generalized Proca theory on flat space does not get renormalized in the decoupling limit. This establishes a clear hierarchy between classical and quantum operators and ensures the validity of Vainshtein screening, as there exists a region characterized by $\alpha_{\rm q,\tilde{q}}\ll 1$ in which classical non-linear interactions dominate over quantum effects and are comparable in size to the kinetic term with $\alpha_{\rm cl}\sim \mathcal{O}(1)$. In conclusion, Generalized Proca theory in the decoupling limit is manifestly protected from large quantum corrections, which renders the effective field theory coefficients technically natural and establishes quantum stability for the full theory in unitary gauge, as described in Sec. \ref{DL_Hierarchy}. Since the operations of taking the decoupling limit and computing quantum corrections are interchangeable, we can directly derive the quantum counterterms that give the largest contribution in unitary gauge from \eqref{ctDL} after substituting $\partial\phi\rightarrow m\,\delta A$, yielding

\begin{equation}
    \mathcal{L}_{\rm flat}^{\rm ct}\sim \partial^4\left(\frac{\partial F_{\delta A}}{\Lambda_3^3}\right)^{2b_2}\left(\frac{m\,\partial\,\delta A}{\Lambda_3^3}\right)^{b_3},
\end{equation}
which demonstrates explicitly that the dangerous high-momentum terms expected from power counting are manifestly absent. Notably, the least suppressed corrections, meaning those with $b_3=0$, appear strictly in gauge-invariant form.
In conclusion, although the presence of a mass scale allows for more irrelevant operators to be generated at the quantum level, these terms are suppressed by additional factors of $\sim m/\Lambda_3$. This directly translates into the statement that their associated ghosts are introduced at energy scales far beyond the effective field theory cutoff where they cannot harm the effective field theory structure. In a next step, we will discuss and generalize these arguments to the case where gravity is turned on.

\subsection{Quantum stability in the presence of gravity}
When adopting the model \eqref{vecgal} as effective field theory expansion for Generalized Proca theory, a large part of the stability analysis parallels the non-renormalization observed on flat space. Consequently, classical operators can only be detuned by couplings with the graviton. To analyse the theory's behaviour in the critical high-energy regime, we apply the decoupling limit to the full model \eqref{vecgal} and observe that only very few terms prevail

\begin{equation}
\label{LfullDL}
    \begin{aligned}
        \mathcal{L}_{\rm DL}\sim \,&(F_{\delta A}^2+(\partial\phi)^2)\left(\frac{\partial^2\phi}{\Lambda_3^3}\right)^{a_3}+(\partial h)^2
        \\[10pt]
        &+\Big((\partial h)(\partial\phi)+h(\partial^2\phi)\Big)\left[\frac{F_{\delta A}}{\Lambda_4^2}+\frac{(\partial h)(\partial\phi)}{\Lambda_4^4}+\frac{h(\partial^2\phi)}{\Lambda_4^4}\right],
    \end{aligned}
\end{equation}
where the precise shape of the individual operators, aside from being constrained by Lorentz invariance, is such that the theory propagates the desired three degrees of freedom. Evidently, the flat-space structure \eqref{LDLflat} is supplemented by the kinetic term of the graviton as well as non-trivial cubic and quartic interactions of the three vector polarizations with the graviton, which entirely stem from expanding the kinetic term $F$ up to fourth order. Crucially, the series stops at quartic order, because higher-point interactions arising from expansions of the coupling functions $G_i$ or metric quantities are additionally suppressed by factors of $\sim\frac{1}{\Lambda_2^2}$, $\sim \frac{m}{\Lambda_2^2}$ or $\sim \frac{1}{M_{\rm pl}}$, causing them to vanish in the decoupling limit. We discuss this crucial point in more detail in Appendix \ref{app4}. From the classical Lagrangian \eqref{LfullDL}, we observe that each vertex brings in a factor of at least $\sim \frac{1}{\Lambda_3^3}$ or $\sim \frac{1}{\Lambda_4^2}$. In combination with the healthy $\sim \frac{1}{p^2}$ behaviour of the propagators, standard power counting arguments together with Lorentz invariance allow to identify the generic structure of a quantum counterterm induced by one-loop 1PI graphs, which consists of the following simple building blocks 

\begin{equation}
\label{Lfullctgen}
\begin{aligned}
    \mathcal{L}_{\rm DL}^{\rm ct}\sim \partial^4
    \left(\frac{\partial F_{\delta A}}{\Lambda_3^3}\right)^{b_1} \left(\frac{\partial F_{\delta A}}{\Lambda_4^3}\right)^{c_1}
    \left(\frac{\partial^2\phi}{\Lambda_3^3}\right)^{b_2}
    \left(\frac{\partial^2\phi}{\Lambda_4^3}\right)^{c_2}\left(\frac{\partial\phi}{\Lambda_4^2}\right)^{c_3}\left(\frac{F_{\delta A}}{\Lambda_4^2}\right)^{c_4} \left(\frac{\partial h}{\Lambda_4^2}\right)^{c_5}
    \left(\frac{h}{\Lambda_4}\right)^{c_6},
\end{aligned}
\end{equation}
where $\sum_{i=1}^2 b_i+\sum_{j=1}^6 c_j=N$ gives the number of external legs for any integers $b_i,\,c_j\geq 0$. Clearly, each counterterm that can be formed from \eqref{Lfullctgen} carries more derivatives than any operator in the classical Lagrangian \eqref{LfullDL}, such that these terms do not get renormalized at one loop order. To highlight the form of the counterterms correcting the 
classical operators \eqref{LfullDL}, we rewrite Eq. \eqref{Lfullctgen} as corrections to the kinetic terms of the massless vector, the scalar and the graviton, as well as the leading interactions with the graviton. Firstly, quantum corrections to the kinetic vector term take the following form

\begin{equation}
\label{delF}
    \begin{aligned}
        \mathcal{L}_{\rm DL}^{\text{ct},\,F_{\delta A}^2}\sim \begin{cases}
            F_{\delta A}^2\left(\frac{\partial}{\Lambda_3}\right)^{4+b_1+\tilde{b}_3}\left(\frac{\partial}{\Lambda_4}\right)^{c_1}\left(\frac{F_{\delta A}}{\Lambda_3^2}\right)^{b_1-2}\left(\frac{F_{\delta A}}{\Lambda_4^2}\right)^{c_1+c_4}\cdot\\[8pt]
            \hspace{25pt}\cdot\left(\frac{\partial^2\phi}{\Lambda_3^3}\right)^{b_2-\tilde{b}_3}
            \left(\frac{\partial^2\phi}{\Lambda_4^3}\right)^{c_2} \left(\frac{\partial\phi}{\Lambda_3^2}\right)^{\tilde{b}_3}
            \left(\frac{\partial\phi}{\Lambda_4^2}\right)^{c_3}\left(\frac{\partial h}{\Lambda_4^2}\right)^{c_5}
            \left(\frac{h}{\Lambda_4}\right)^{c_6},\hspace{16pt} b_1\geq 2\\[13pt]
            F_{\delta A}^2\left(\frac{\partial}{\Lambda_3}\right)^{b_1+\tilde{b}_3}\left(\frac{\partial}{\Lambda_4}\right)^{4+c_1}\left(\frac{F_{\delta A}}{\Lambda_3^2}\right)^{b_1}\left(\frac{F_{\delta A}}{\Lambda_4^2}\right)^{c_1+c_4-2}\Big[\dots\Big]
            ,\hspace{7pt} c_1\,\,\text{or}\,\,c_4\geq 2\\[13pt]
            F_{\delta A}^2\left(\frac{\partial}{\Lambda_3}\right)^{2+b_1+\tilde{b}_3}\left(\frac{\partial}{\Lambda_4}\right)^{c_1+c_4+1}\left(\frac{F_{\delta A}}{\Lambda_4^2}\right)^{c_1+c_4-1}\Big[\dots\Big],\hspace{6pt} b_1=1,\,c_{1,4}\leq 1,\\
            \hspace{246pt}c_1\,\,\text{or}\,\,c_4=1\\[3pt]
            F_{\delta A}^2\left(\frac{\partial}{\Lambda_3}\right)^{\tilde{b}_3}\left(\frac{\partial}{\Lambda_4}\right)^{3+c_1+c_4}\Big[\dots\Big],\hspace{97pt} b_1=0,\, c_{1,4}=1
        \end{cases}
    \end{aligned}
\end{equation}
after suitable integration by parts, where $\tilde{b}_3\leq b_2$ and $[\dots]$ indicates the factors of $\sim \partial^2\phi,\,\partial\phi,\,\partial h,\, h$ we give explicitly in the second line but then suppressed for better readability. Contrary to flat space\footnotemark\footnotetext{\vtop{\hsize=\linewidth\noindent In the classical interactions on flat space, the field strength tensor $F_{\delta A}$ always couples to factors of $\sim \partial^2\phi$, such that the scalar external legs in quantum-induced scalar-vector operators always have two derivatives acting on them. Moreover, pure scalar counterterms can always be integrated by parts to feature only two factors of $\sim \partial\phi$.}}, the number of external scalar legs with one derivative $\sim \partial \phi$ is not restricted to two (a remnant from Galileon theories), as the Proca field can linearly couple to the graviton (see Eq. \eqref{LfullDL}), which allows for a richer structure of the induced quantum counterterms with multiple scalar legs of the type $\sim \partial\phi$ and $\sim\partial^2\phi$ \cite{Nicolis:2004qq}.

Secondly, corrections to the scalar kinetic term can be written as

\begin{equation}
\label{delphi}
\begin{aligned}
        \mathcal{L}_{\rm DL}^{\text{ct},\,(\partial\phi)^2}\sim \begin{cases}
            (\partial\phi)^2\left(\frac{\partial}{\Lambda_3}\right)^{6+b_1+\tilde{b}_3}\left(\frac{\partial}{\Lambda_4}\right)^{c_1}\left(\frac{\partial^2\phi}{\Lambda_3^3}\right)^{b_2-\tilde{b}_3-2}\left(\frac{\partial^2\phi}{\Lambda_4^3}\right)^{c_2} \left(\frac{\partial\phi}{\Lambda_3^2}\right)^{\tilde{b}_3}
            \left(\frac{\partial\phi}{\Lambda_4^2}\right)^{c_3}\cdot\\[8pt]
            \hspace{80pt}\cdot\left(\frac{F_{\delta A}}{\Lambda_3^2}\right)^{b_1}\left(\frac{F_{\delta A}}{\Lambda_4^2}\right)^{c_1+c_4}
            \left(\frac{\partial h}{\Lambda_4^2}\right)^{c_5}
            \left(\frac{h}{\Lambda_4}\right)^{c_6},\hspace{81pt} b_2\geq 2\\[13pt]
            (\partial\phi)^2\left(\frac{\partial}{\Lambda_3}\right)^{b_1+\tilde{b}_3}\left(\frac{\partial}{\Lambda_4}\right)^{6+c_1}\left(\frac{\partial^2\phi}{\Lambda_3^3}\right)^{b_2-\tilde{b}_3}
            \left(\frac{\partial^2\phi}{\Lambda_4^3}\right)^{c_2-2} \left(\frac{\partial\phi}{\Lambda_3^2}\right)^{\tilde{b}_3}
            \left(\frac{\partial\phi}{\Lambda_4^2}\right)^{c_3}\Big[\dots\Big],\hspace{5pt} c_2\geq 2\\[13pt]
            (\partial\phi)^2\left(\frac{\partial}{\Lambda_3}\right)^{b_1+\tilde{b}_3}\left(\frac{\partial}{\Lambda_4}\right)^{4+c_1}\left(\frac{\partial^2\phi}{\Lambda_3^3}\right)^{b_2-\tilde{b}_3}
            \left(\frac{\partial^2\phi}{\Lambda_4^3}\right)^{c_2} 
            \left(\frac{\partial\phi}{\Lambda_3^2}\right)^{b_3}
            \left(\frac{\partial\phi}{\Lambda_4^2}\right)^{c_3-2}\Big[\dots\Big],\hspace{5pt} c_3\geq 2
        \end{cases}
    \end{aligned}
\end{equation}
after integrating by parts, where $\tilde{b}_3\leq b_2$ and $[\dots]$ again denotes the suppressed factors of $\sim F_{\delta A},\,\partial h,\,h$ from the second line. In particular, corrections to the Proca kinetic term in \eqref{ctprocakin} are now represented by the parameter choices $\{b_{1,2}=0,c_{1,2,4,5,6}=0,c_3=2\}$ and $\{b_{1,2}=0,c_{1,2,3,5,6}=0,c_4=2\}$. Thirdly, corrections to the graviton's kinetic term become

\begin{equation}
\label{delh}
\begin{aligned}
        \mathcal{L}_{\rm DL}^{\text{ct},\,(\partial h)^2}\sim \begin{cases}
            (\partial h)^2\left(\frac{\partial}{\Lambda_3}\right)^{b_1+\tilde{b}_3}\left(\frac{\partial}{\Lambda_4}\right)^{4+c_1}\left(\frac{\partial h}{\Lambda_4^2}\right)^{c_5-2}
            \left(\frac{h}{\Lambda_4}\right)^{c_6}\cdot\\[10pt]
            \hspace{9pt}\cdot\left(\frac{F_{\delta A}}{\Lambda_3^2}\right)^{b_1}\left(\frac{F_{\delta A}}{\Lambda_4^2}\right)^{c_1+c_4}
            \left(\frac{\partial^2\phi}{\Lambda_3^3}\right)^{b_2-\tilde{b}_3}\left(\frac{\partial^2\phi}{\Lambda_4^3}\right)^{c_2} \left(\frac{\partial\phi}{\Lambda_3^2}\right)^{\tilde{b}_3}
            \left(\frac{\partial\phi}{\Lambda_4^2}\right)^{c_3},\hspace{5pt} c_5\geq 2\\[15pt]
            (\partial h)^2\left(\frac{\partial}{\Lambda_3}\right)^{b_1+\tilde{b}_3}\left(\frac{\partial}{\Lambda_4}\right)^{2+c_1}
            \left(\frac{\partial h}{\Lambda_4^2}\right)^{c_5}
            \left(\frac{h}{\Lambda_4}\right)^{c_6-2}\Big[\dots\Big],\hspace{47pt}c_6\geq 2
        \end{cases}
    \end{aligned}
\end{equation}
after integration by parts, with $\tilde{b}_3\leq b_2$ and $[\dots]$ abbreviating all terms $\sim F_{\delta A},\,\partial^2\phi,\,\partial\phi$ given in the second line. The counterterm to the graviton's kinetic term in \eqref{ctgravkin} is now captured by setting $\{b_{1,2}=0,c_{1,2,3,4,6}=0,c_5=2\}$. Overall, Eqs. \eqref{delF}-\eqref{delh} allow us to define alternative expansion parameters, which we denote as

\begin{equation}
\begin{aligned}
\label{betas}
    &\beta_{\rm cl,\partial\phi}\equiv \frac{\partial\phi}{\Lambda_4^2},\quad\beta_{\text{cl},\partial h}\equiv\frac{\partial h}{\Lambda_4^2},\quad\beta_{\text{cl}, F_{\delta A}}\equiv \frac{F_{\delta A}}{\Lambda_4^2},\quad\beta_{\text{cl},\partial^2\phi}\equiv\frac{\partial^2\phi}{\Lambda_4^3},\quad\beta_{\text{cl},h}\equiv\frac{h}{\Lambda_4},\\[10pt]
    &\tilde{\alpha}_{\rm q}\equiv \frac{\partial}{\Lambda_3},\hspace{25pt}\tilde{\alpha}_{\rm \tilde{q}}\equiv \frac{F_{\delta A}}{\Lambda_3^2},\hspace{20pt} \tilde{\alpha}_{\rm \tilde{q},new}\equiv \frac{\partial\phi}{\Lambda_3^2},\hspace{15pt} \beta_{\rm q}\equiv \frac{\partial}{\Lambda_4},
\end{aligned}
\end{equation}
in close analogy to their siblings on flat space. Importantly, since the Proca field can linearly couple to the graviton, the expansion parameters $\alpha_{\rm q},\, \alpha_{\rm \tilde{q}}$ \eqref{alphas} on flat space need to be slightly modified according to \eqref{betas} and we have to newly introduce the parameter $\tilde{\alpha}_{\rm \tilde{q},new}$ to accommodate quantum interactions with either external scalar legs $\sim \partial\phi$ \footnote{Note that these external factors were not present on flat space, as the field strength tensor always couples to factors of $\sim \partial^2 \phi$ in the classical vertices, such that the scalar external legs in quantum-induced scalar-vector interactions always have two derivatives acting on them. On the other hand, pure scalar counterterms can always be integrated by parts to feature only two factors of $\sim (\partial\phi)^2$ \cite{Nicolis:2004qq}.} or an uneven number of gauge-preserving external legs $F_{\delta A}$.

We can generalize the above analysis to higher loops in a straightforward way: including an additional loop in a 1PI diagram while simultaneously keeping the same number of external fields requires the addition of either another vertex or another internal leg, which both brings in a new factor of $\sim 1/\Lambda_3^3$ or $\sim1/\Lambda_4^2$. As we keep the number of external legs fixed, the only way in which the dimension of the resulting quantum operator remains consistent is through adding more derivatives that cancel these new factors. Therefore, higher loop diagrams can easily be accounted for by additional factors of $\beta_q$. At that point, the establishment of quantum stability in the presence of gravity follows the same logic as on flat space: We write the complete effective field theory Lagrangian as an expansion in the classical and quantum parameters \eqref{alphas}, \eqref{betas}

\begin{equation}
\begin{aligned}
    \mathcal{L}_{\rm DL}&\sim(F_{\delta A}^2+(\partial\phi)^2)\,\alpha_{\rm cl}^{a_3}+(\partial h)^2+[h(\partial^2\phi)+(\partial h)(\partial\phi)](\beta_{\text{cl},F_{\delta A}}+\beta_{\text{cl},\partial h}\,\beta_{\text{cl},\partial \phi}+\beta_{\text{cl}, h}\,\beta_{\text{cl},\partial^2 \phi})\\[10pt]
    &+(F_{\delta A}^2+(\partial\phi)^2+(\partial h)^2)(\tilde{\alpha}_{\rm q}^{2}+\beta_{\rm q}^{2})\tilde{\alpha}_{\rm q}^{m_0}\,\beta_{\rm q}^{m_1}\,\tilde{\alpha}_{\rm \tilde{q}}^{m_2}\,\tilde{\alpha}_{\rm \tilde{q},new}^{m_3} \alpha_{\rm cl}^{m_4}\,\beta_{\text{cl},F_{\delta A}}^{m_5}\,\beta_{\text{cl},\partial h}^{m_6}\,\beta_{\rm cl,\partial\phi}^{m_7}\,\beta_{\text{cl}, h}^{m_8}\,\beta_{\rm cl,\partial^2\phi}^{m_9},\\[10pt]
    &+[h(\partial^2\phi)+(\partial h)(\partial\phi)](\tilde{\alpha}_{\rm q}^{2}+\beta_{\rm q}^{2})\tilde{\alpha}_{\rm q}^{n_0}\,\beta_{\rm q}^{n_1}\,\tilde{\alpha}_{\rm \tilde{q}}^{n_2}\,\tilde{\alpha}_{\rm \tilde{q},new}^{n_3} \alpha_{\rm cl}^{n_4}\,\beta_{\text{cl},F_{\delta A}}^{n_5}\,\beta_{\text{cl},\partial h}^{n_6}\,\beta_{\rm cl,\partial\phi}^{n_7}\,\beta_{\text{cl}, h}^{n_8}\,\beta_{\rm cl,\partial^2\phi}^{n_9},
\end{aligned}
\end{equation}
where $a_3,m_i,n_j\geq 0$ and at least one $n_{i\geq 5}>0$, as these terms give corrections to the classical Proca-graviton interactions\footnote{We could have obtained the expansion parameters \eqref{betas} analogously from comparing the generic counterterm \eqref{Lfullctgen} to the surviving classical three- and four-point interactions in \eqref{LfullDL} instead of the kinetic terms.}. For example, the parameter choices $\{m_1=2,m_4=2\}$, $\{m_0=4,m_6=2\}$, $\{n_0=1,n_1=3,n_4=1,n_8=1\}$, $\{m_0=4,m_1=2,m_8=2\}$, $\{m_1=4,m_8=2\}$, $\{m_1=2,m_6=2\}$, $\{m_1=2,m_7=2\}$, $\{n_1=2,n_6=1,n_7=1\}$ and $\{n_1=2,n_8=1,n_9=1\}$, while all other $m_i,n_i=0$, contain all one-loop counterterms to the quartic vertex $\sim h^2\delta A^2$ in \eqref{cthhAA}, while $\{n_1=2,n_4=1,n_5=1\}$ and $\{n_0=1,n_1=4,n_2=1,n_8=1\}$ capture the ones for $\sim \delta A^3 h$ in \eqref{cthAAA}. Similarly, the one-loop corrections to the surviving cubic interactions are represented by $\{m_1=2,m_4=1\}$ in the pure Proca sector \eqref{ctAAA} and by $\{n_1=2,n_5=1\}$ as well as $\{n_0=1,n_1=3,n_8=1\}$, $\{m_1=2,m_4=1\}$ in the mixed Proca-graviton sectors \eqref{ctAAh} and \eqref{LhhADL}, respectively.
Evidently, the quantum-induced operators are accompanied by factors of $\tilde{\alpha}_{q,\tilde{q}},\,\beta_{q}$, implying that every loop contribution necessarily comes with higher derivative orders compared to the classical terms. This generates a clear hierarchy of scales between classical and quantum operators and, moreover, establishes the manifest non-renormalization of classical terms, even in the presence of gravity. Therefore, the Vainshtein screening mechanism remains intact, as there exists a regime with large classical non-linearities $\alpha_{\rm cl},\beta_{\rm cl}\sim \mathcal{O}(1)$ while, simultaneously, quantum effects are subdominant $\tilde{\alpha}_{\rm q,\tilde{q}},\beta_{q}\ll 1$. Overall, this marks radiative stability of Generalized Proca theories in the decoupling limit, which directly extends to radiative stability of the full effective field theory, as pointed out in Sec. \ref{DL_Hierarchy}. We conclude that Generalized Proca theory maintains its essential classical features when quantum effects and gravity are included, and thus can indeed serve as theoretically viable IR modification of gravity when viewed as an effective field theory. 

Notice again that the commutativity of taking the decoupling limit and computing quantum corrections allows to extract the least suppressed counterterms in the original unitary-gauge formulation. Upon making the replacement $\partial\phi\rightarrow m\,\delta A$ in Eq. \eqref{Lfullctgen}, we obtain 
\small{
\begin{equation}
\begin{aligned}
    \mathcal{L}^{\rm ct}\sim \partial^4
    \left(\frac{\partial F_{\delta A}}{\Lambda_3^3}\right)^{b_1}
    \left(\frac{m\,\partial\, \delta A}{\Lambda_3^3}\right)^{b_2}
    \left(\frac{\partial F_{\delta A}}{\Lambda_4^3}\right)^{c_1}
    \left(\frac{m\,\partial\, \delta A}{\Lambda_4^3}\right)^{c_2}
    \left(\frac{m\,\delta A}{\Lambda_4^2}\right)^{c_3}\left(\frac{F_{\delta A}}{\Lambda_4^2}\right)^{c_4}
    \left(\frac{\partial h}{\Lambda_4^2}\right)^{c_5}
    \left(\frac{h}{\Lambda_4}\right)^{c_6},
\end{aligned}
\end{equation}}
where, again, the leading counterterms are exactly the gauge-invariant ones with $b_{2}= c_{2,3}=0$, possibly dressed by factors involving the graviton. All other contributions in unitary gauge are further suppressed by additional powers of $\sim m/\Lambda_3^3$, $\sim m/\Lambda_4^2$ and $\sim m/\Lambda_4^3$, and hence do not spoil the effective field theory hierarchy.

Taken together, these considerations show that the non-renormalization theorem is not a consequence of special cancellations at low loop order, but rather a structural property of Generalized Proca effective field theories coupled to gravity. The decoupling limit isolates the maximally sensitive sector of the theory and exhausts the space of admissible counterterms below the cutoff, ensuring radiative stability independently of perturbative order.




\section{Conclusion}
\label{conclusion}
Within the modern effective field theory paradigm, classical theories of gravity and their modifications are understood as low-energy approximations to an underlying quantum description. From this perspective, theoretical consistency
requires not only a well-defined classical dynamics but also stability of the theory under radiative corrections within its regime of validity. In models with derivative self-interactions, this requirement is particularly stringent, as the
nonlinear structures responsible for screening mechanisms are potentially sensitive to quantum-induced operators.

In this work, we have analyzed the structure of radiative corrections in Generalized Proca effective field theories coupled to dynamical gravity. Treating these models as local, diffeomorphism-covariant functionals and working within a
perturbative expansion about flat spacetime, we investigated the fate of the classical ghost-free interactions under quantum corrections. Explicit one-loop computations of two- and three-point functions, combined with a systematic power
counting analysis in dimensional regularization, show that potentially dangerous contributions either cancel or appear in higher-derivative, gauge-invariant combinations.

A central role in this analysis is played by the decoupling limit, which isolates the dominant high-energy interactions of the theory in a controlled scaling regime. Reformulating the theory via the Stückelberg decomposition makes the
separation between transverse and longitudinal modes explicit and allows for a transparent classification of quantum corrections at the level of local operators. In this limit, the propagating degrees of freedom exhibit standard
ultraviolet behavior, and all radiatively generated counterterms are suppressed by the intrinsic strong-coupling scales of the effective theory.

These results culminate in a non-renormalization theorem for Generalized Proca theories coupled to gravity: the classical interaction terms responsible for the constraint structure and the absence of Ostrogradsky instabilities are not
renormalized by quantum effects within the effective field theory regime. All quantum-induced local operators necessarily involve additional derivatives and enter at scales parametrically above the cutoff. Consequently, both the low-energy
dynamics and the number of propagating degrees of freedom remain invariant under radiative corrections to all orders in perturbation theory.

The analysis presented here relies on an expansion about backgrounds with vanishing vector expectation value, for which the kinetic sectors of the metric and vector field diagonalize at quadratic order. Allowing for nontrivial background vector
configurations is expected to modify propagators and interaction vertices through background-dependent insertions but does not alter the underlying scaling structure that controls the decoupling limit and the hierarchy of operators.
A detailed treatment of such backgrounds is left for future work.

Overall, our results demonstrate that radiative stability in Generalized Proca theories is a structural consequence of the effective field theory hierarchy and the decoupling limit, rather than an artifact of low-order perturbative
calculations. This places Generalized Proca models on a firm footing as consistent quantum effective field theories of vector degrees of freedom coupled to gravity.

\acknowledgments
We acknowlegde the use of the xAct, FeynCalc, FeynArts and X packages for Mathematica. L.H. is supported by funding from the European Research Council (ERC) under the European Unions Horizon 2020 research and innovation programme grant agreement No 801781. L.H. further acknowledges support from the Deutsche Forschungsgemeinschaft (DFG, German Research Foundation) under Germany’s Excellence Strategy EXC 2181/1 - 390900948 (the Heidelberg STRUCTURES Excellence Cluster).

\appendix
\section{Decoupling scales}
\label{app2}
In this appendix, we show the emergence of the decoupling scales that we mentioned in Section~\ref{DL_Hierarchy}. After performing the Stückelberg transformation for the cubic and quartic Lagrangians \eqref{Lcubic}-\eqref{Lquartic}, we find schematically

\begin{equation}
\begin{aligned}
    \mathcal{L}^{(3)}_{2, \text{Stb}} &\sim \frac{1}{M_{\rm pl}m}h\,F_{\delta A}(\partial^2\phi)+\frac{1}{M_{\rm pl}m}(\partial h)F_{\delta A}(\partial\phi)+\frac{1}{M_{\rm pl}}hF_{\delta A}^2+\frac{1}{M_{\rm pl}}h(\partial\,\delta A)^2\\[5pt]
    &\quad+\frac{1}{M_{\rm pl}}h(\partial \phi)^2+\frac{1}{M_{\rm pl}}(\partial h)\delta A(\partial\,\delta A)
    +\frac{m}{M_{\rm pl}}h\,\delta\,A(\partial\phi)+\frac{m^2}{M_{\rm pl}}h\,\delta A^2, 
    \end{aligned}
\end{equation}
\begin{equation}
    \begin{aligned}
    \mathcal{L}^{(3)}_{3, \text{Stb}} &\sim \frac{c_3}{\Lambda_2^2 m}\, (\partial \phi) ^2(\partial^2 \phi)+
    \frac{c_3}{\Lambda_2^2}\, \delta A( \partial \phi)( \partial^2 \phi)
    + \frac{c_3}{\Lambda_2^2}\, (\partial \phi)( \partial \phi)( \partial\, \delta A)\\[5pt]
    &\quad
    + \frac{c_3 m^2}{\Lambda_2^2}\, \delta A^2 ( \partial\, \delta A) 
    + \frac{c_3 m}{\Lambda_2^2}\, \delta A^2 (\partial^2 \phi),
\end{aligned}
\end{equation}

\begin{equation}
    \begin{aligned}
        \mathcal{L}^{(4)}_{2, \text{Stb}} &\sim \frac{1}{M_{\rm pl}^2m^2}h^2(\partial^2\phi)^2+\frac{1}{M_{\rm pl}^2m^2}h(\partial^2\phi)(\partial h)(\partial \phi)+\frac{1}{M_{\rm pl}^2m^2}(\partial h)^2(\partial\phi)^2\\[5pt]
        &\quad +\frac{1}{M_{\rm pl}^2m}h\,F_{\delta A}(\partial h)(\partial \phi)+\frac{1}{M_{\rm pl}^2m}h^2(\partial\,\delta A)(\partial^2\phi)+\frac{1}{M_{\rm pl}^2m}h(\partial\,\delta A)(\partial h)(\partial\phi)\\[5pt]
        &\quad +\frac{1}{M_{\rm pl}^2m}(\partial h)^2\delta A(\partial\phi)+\frac{1}{M_{\rm pl}^2m}h\,\delta A(\partial h)(\partial^2\phi)+\frac{1}{M_{\rm pl}^2}h^2F_{\delta A}^2+\frac{1}{M_{\rm pl}^2}h^2(\partial\,\delta A)^2\\[5pt]
        &\quad +\frac{1}{M_{\rm pl}^2}h^2(\partial\phi)^2+\frac{1}{M_{\rm pl}^2}h(\partial\,\delta A)(\partial h)\delta A+\frac{1}{M_{\rm pl}^2}(\partial h)^2\delta A^2+\frac{m}{M_{\rm pl}^2}h^2\delta A(\partial\phi)\\[5pt]
        &\quad+\frac{m^2}{M_{\rm pl}^2}h^2\delta A^2,
    \end{aligned}
\end{equation}
\begin{equation}
    \begin{aligned}
        \mathcal{L}^{(4)}_{3, \text{Stb}} &\sim \frac{c_3}{M_{\rm pl}\Lambda_2^2 m}h\,(\partial\phi)^2(\partial^2\phi)+\frac{c_3}{M_{\rm pl}\Lambda_2^2 m}(\partial h)(\partial\phi)^3+\frac{c_3}{M_{\rm pl}\Lambda_2^2}h(\partial\phi)^2(\partial\,\delta A)\\[5pt]
        &\quad+\frac{c_3}{M_{\rm pl}\Lambda_2^2}h\,\delta A(\partial\phi)(\partial^2\phi)+\frac{c_3}{M_{\rm pl}\Lambda_2^2}\,\delta A(\partial\phi)^2(\partial h)+\frac{c_3\,m}{M_{\rm pl}\Lambda_2^2}\,\delta A^2(\partial h)(\partial\phi)\\[5pt]
        &\quad+\frac{c_3\,m}{M_{\rm pl}\Lambda_2^2}\,h\,\delta A(\partial^2\phi)+\frac{c_3\,m}{M_{\rm pl}\Lambda_2^2}\,h\,\delta A(\partial\phi)(\partial\,\delta A)+\frac{c_3\,m^2}{M_{\rm pl}\Lambda_2^2}\,h\,\delta A^2(\partial\,\delta A)\\[5pt]
        &\quad +\frac{c_3\,m^2}{M_{\rm pl}\Lambda_2^2}\,\delta A^3(\partial h).
    \end{aligned}
\end{equation}
We observe that the Stückelberg-ed theory introduces interaction terms involving scalar, vector, and tensor modes that are controlled by distinct energy scales. These scales can be extracted by identifying the leading operators in the high-energy regime. The two relevant suppression scales that emerge from this expansion are given by
\begin{equation}
{\Lambda}_4 = \left( m M_{\rm pl} \right)^{1/2}, \qquad
\Lambda_3 = \left( \Lambda_2^2 m \right)^{1/3},
\end{equation}
denoting the strong coupling scale of graviton-Proca and pure scalar interactions, respectively. By focusing on the physics at these two energy scales in the decoupling limit \eqref{DL}, most mixing terms above, in particular those with the transverse vector modes that do not appear in gauge-invariant form, disappear. The remaining terms for the quartic Lagrangian are
\begin{equation}
    \begin{aligned}
        \mathcal{L}^{(4)}_{\rm{DL}}=& -\frac{2}{\Lambda_4^4}h_\alpha{}^\delta h^{\alpha\beta}(\partial_\gamma\partial_\beta\phi)(\partial^\gamma\partial_\delta\phi)
        +\frac{2}{\Lambda_4^4}(\partial_\gamma h_\beta{}^\delta)(\partial_\delta h_{\alpha}{}^\gamma)(\partial^\alpha\phi)(\partial^\beta\phi)\\
        &+\frac{4}{\Lambda_4^4}h^{\beta\gamma}(\partial^\alpha\phi)(\partial_\gamma h_\alpha{}^\delta)(\partial_\delta\partial_\beta\phi)+\frac{2}{\Lambda_4^4}h^{\alpha\beta}h^{\gamma\delta}(\partial_\gamma\partial_\alpha\phi)(\partial_\delta\partial_\beta\phi)\\
        &-\frac{4}{\Lambda_4^4}h^{\beta\gamma}(\partial^\alpha\phi)(\partial_\delta\partial_\beta\phi)(\partial^\delta h_{\alpha\gamma})-\frac{2}{\Lambda_4^4}(\partial^\alpha\phi)(\partial^\beta\phi)(\partial_\delta h_\alpha{}^\gamma)(\partial^\delta h_{\beta\gamma}).
    \end{aligned}
\end{equation}

\section{Exact expressions for the interactions at higher orders in perturbations}
\label{app1}

Expanding the Generalized Proca Lagrangians \eqref{gpLagr} in the minimal model \eqref{minmodel} to fourth order in perturbations yields exactly 
\begin{align}
    \mathcal{L}^{(4)}_{2} &= 
            -\frac{m^2 }{2 M_{\rm pl}^2}\, h_{\beta\gamma} \,h^{\beta\gamma} \,\delta A_{\alpha}\, \delta A^{\alpha}
            + \frac{m^2 }{4M_{\rm pl}^2}\, h^2 \,\delta A_{\alpha} \,\delta A^{\alpha}
            + \frac{ m^2 }{M_{\rm pl}^2}\, h_{\alpha\beta}\, h\, \delta A^{\alpha}\, \delta A^{\beta} \\
            &\quad + \frac{1}{4M_{\rm pl}^2}\, h_{\gamma\delta} \,h^{\gamma\delta}\,\,F_{\beta\alpha}^{\delta A}\,F^{\beta\alpha}_{\delta A}- \frac{1}{8M_{\rm pl}^2}\, h^2 \,F_{\beta\alpha}^{\delta A}\,F^{\beta\alpha}_{\delta A}
            \notag \\
            &\quad - \frac{2}{M_{\rm pl}^2}\, h^{\delta}{}_{\beta} \,h_{\gamma\delta}\,( \partial^{\beta} \delta A^{\alpha} )\,(\partial^{\gamma} \delta A_{\alpha}) + \frac{2}{M_{\rm pl}^2}\, h_{\alpha\gamma}\, h_{\beta\delta}\,( \partial^{\beta} \delta A^{\alpha} )\,(\partial^{\delta} \delta A^{\gamma}),\notag \\
            &\quad  - \frac{1}{M_{\rm pl}^2}\, h_{\alpha\gamma}\, h \,(\partial_{\beta} \delta A^{\gamma})\,( \partial^{\beta} \delta A^{\alpha})
            + \frac{1}{M_{\rm pl}^2}\, h_{\beta\gamma} \,h \,(\partial^{\beta} \delta A^{\alpha})\,( \partial^{\gamma} \delta A_{\alpha})\notag \\
            &\quad
            + \frac{2}{M_{\rm pl}^2}\, h \,\delta A^{\alpha} \,(\partial_{\beta} h_{\alpha\gamma}) \,(\partial^{\gamma} \delta A^{\beta}) - \frac{2}{M_{\rm pl}^2}\, h\, \delta A^{\alpha} \,(\partial_{\gamma} h_{\alpha\beta})\,( \partial^{\gamma} \delta A^{\beta})  \notag \\
            &\quad + \frac{4}{M_{\rm pl}^2}\, h^{\delta}{}_{\gamma}\, \delta A^{\alpha} \,(\partial^{\gamma} \delta A^{\beta})\,( \partial_{\delta} h_{\alpha\beta}) - \frac{4}{M_{\rm pl}^2}\, h^{\delta}{}_{\gamma}\, \delta A^{\alpha} \,(\partial_{\beta} h_{\alpha\delta} )\,(\partial^{\gamma} \delta A^{\beta})\notag\\
            &\quad + \frac{2}{M_{\rm pl}^2}\, \delta A^{\alpha}\, \delta A^{\beta} \,(\partial_{\gamma} h_{\beta\delta} )\,(\partial^{\delta} h_{\alpha}{}^{\gamma}) - \frac{2}{M_{\rm pl}^2}\, \delta A^{\alpha}\, \delta A^{\beta}\,( \partial_{\delta} h_{\beta\gamma} )\,(\partial^{\delta} h_{\alpha}{}^{\gamma}),\\[10pt]
        \mathcal{L}^{(4)}_{3} &= 
            \frac{m^2}{2M_{\rm pl}\Lambda_2^2}\, c_3\, \delta A_{\beta}\, \delta A^{\beta}\, \delta A^{\alpha}\, (\partial_{\alpha} h)
            +  \frac{m^2}{2M_{\rm pl}\Lambda_2^2}\, c_3\, h\, \delta A_{\beta}\, \delta A^{\beta}\, (\partial_{\alpha} \delta A^{\alpha})\notag\\
            &\quad  +  \frac{m^2}{M_{\rm pl}\Lambda_2^2}\, c_3\, h_{\beta\alpha}\, \delta A^{\alpha}\, \delta A^{\beta}\, (\partial_{\gamma} \delta A^{\gamma}).
\end{align}

\section{One-loop contributions to the reduced matrix elements}
\label{app3}
The full contribution from the one-loop tadpole and bubble diagram to the reduced matrix element for the graviton involving two external momenta reads
\begin{equation}
    \begin{aligned}
        \mathcal{M}_{2,h}^{\rm{div}}\supset 
        & \frac{\epsilon^{\mu\nu}_p\,\epsilon^{\alpha\beta}_{-p}}{1920\, \pi^2\varepsilon  \,M_{\rm{pl}}^2}\Big\{ p^{2} \Big[900\, m^2 \eta_{\alpha \mu } \eta_{\beta \nu }-120\, m^2 \eta_{\alpha \nu } \eta_{\beta \mu }+2 \eta_{\mu \nu} \big(p_{\alpha } p_{\beta }-30\, m^2 \eta_{\alpha \beta }\big)\\[5pt]
        &+67 \,p_{\beta } p_{\nu } \eta_{\alpha \mu }-93\, p_{\alpha } p_{\nu } \eta_{\beta \mu }+p_{\mu } \big(2\, p_{\nu } \eta_{\alpha \beta}
        -93\, p_{\beta } \eta_{\alpha \nu}
        +67\, p_{\alpha } \eta_{\beta \nu }\big)\Big]\\[5pt]
        &+\frac{p^{4}}{m^2} \Big[m^2 \big(63\, \eta_{\alpha \nu } \eta_{\beta \mu }-37\,\eta_{\alpha \mu } \eta_{\beta \nu }-2\, \eta_{\alpha \beta } \eta_{\mu \nu }\big)\\[5pt]
        &+10\, p_{\beta } \big(p_{\mu } \eta_{\alpha\nu }-p_{\nu } \eta_{\alpha \mu }\big)
        +10\, p_{\alpha } \big(p_{\nu } \eta_{\beta \mu }-p_{\mu } \eta_{\beta \nu}\big)\Big]\\[5pt]
        &+10\, m^2 \Big[m^2 \big(13\, \eta_{\alpha \nu } \eta_{\beta \mu }+49\, \eta_{\alpha \mu } \eta_{\beta \nu }-47 \eta_{\alpha \beta } \eta_{\mu \nu }\big)\\[5pt]
        &-99\, p_{\beta } p_{\nu } \eta_{\alpha \mu }+3\, p_{\mu } \big(2 p_{\nu }\eta_{\alpha \beta }+7\, p_{\beta } \eta_{\alpha \nu }\big)\Big]+6\, p_{\alpha } \Big[5\, m^2 \big(7\, p_{\nu }\eta_{\beta \mu }-9 \,p_{\mu } \eta_{\beta \nu }\big)\\[5pt]
        &+2\, p_{\beta } \big(5\, m^2 \eta_{\mu \nu }+2 \, p_{\mu } p_{\nu}\big)\Big]+5\, \frac{p^{6}}{m^2} \big(\eta_{\alpha \mu } \eta_{\beta \nu }-\eta_{\alpha \nu } \eta_{\beta \mu }\big)\Big\}.
    \end{aligned}
\end{equation}

\section{Schematic perturbative expansion of Generalized Proca theory}
\label{app4}
The underlying structure of Generalized Proca theory was crucial in establishing quantum stability in the presence of the graviton. Here, we give the schematic interactions arising in the perturbative expansion of the full covariant expansion around the background configuration \eqref{bgconfig}. At cubic order, the Lagrangians read

\begin{equation}
    \begin{aligned}
        &\mathcal{L}_2^{(3)}\sim \frac{1}{\sqrt{2 \bar{G}_4}}\left[h\,F_{\delta A}^2+h(\partial\,\delta A)^2+(\partial\,\delta A)(\partial h)\delta A+\frac{\bar{G}_{2,X}}{\bar{G}_{2,F}}\,h\,\delta A^2\right],\\[10pt]
        &\mathcal{L}_3^{(3)}\sim \frac{\bar{G}_{3,X}}{\bar{G}_{2,F}^{3/2}}\,\delta A^2(\partial \,\delta A),\\[10pt]
        &\mathcal{L}_4^{(3)}\sim \frac{1}{\sqrt{2\bar{G}_4}}\big[h\,(\partial h)^2+h^2(\partial^2 h)\big],\\[10pt]
        &\mathcal{L}_5^{(3)}\sim \frac{\bar{G}_5}{2\bar{G}_4\,\bar{G}_{2,F}^{1/2}}\left[h\,(\partial^2h)(\partial\,\delta A)+(\partial h)^2(\partial\,\delta A)\right]+\frac{\bar{G}_{5,X}}{\bar{G}_{2,F}^{3/2}}\,(\partial\,\delta A)^3+\frac{\bar{g}_5}{\bar{G}_{2,F}^{3/2}}\,\tilde{F}_{\delta A}^2(\partial\,\delta A),\\[10pt]
        &\mathcal{L}_6^{(3)}\sim\frac{\bar{G}_6}{\sqrt{2\bar{G}_4}\,\bar{G}_{2,F}}\,(\partial^2h)(\partial\,\delta A)^2,
    \end{aligned}
\end{equation}
where the graviton and the vector field were canonically normalized according to \eqref{can}. At quartic order, we find
\begin{equation}
    \begin{aligned}
        &\mathcal{L}_2^{(4)}\sim \frac{1}{2 \bar{G}_4}\left[h^2\,F_{\delta A}^2+h^2(\partial\,\delta A)^2+h(\partial\,\delta A)(\partial h)\delta A+(\partial h)^2\delta A^2
        +\frac{\bar{G}_{2,X}}{\bar{G}_{2,F}}\,h^2\delta A^2\right]\\[5pt]
        &\qquad\quad+\frac{1}{\bar{G}_{2,F}^2}\left[\bar{G}_{2,FF}\,F_{\delta A}^4+\bar{G}_{2,Y}\,F_{\delta A}^2\delta A^2+\bar{G}_{2,XF}\,F_{\delta A}^2\delta A^2+\bar{G}_{2,XX}\,\delta A^4\right],\\[10pt]
        &\mathcal{L}_3^{(4)}\sim 
        \frac{\bar{G}_{3,X}}{\sqrt{2\bar{G}_4}\bar{G}_{2,F}^{3/2}}\,h\,\delta A^2(\partial \,\delta A),\\[10pt]
        &\mathcal{L}_4^{(4)}\sim \frac{1}{2\bar{G}_4}\left[h^2\,(\partial h)^2
        +\frac{\bar{G}_{4,X}}{\bar{G}_{2,F}}\Big\{\delta A^2h(\partial^2 h)+\delta A^2(\partial h)^2+h^2(\partial\,\delta A)^2+h(\partial\,\delta A)(\partial h)\delta A\Big\}\right.\\[5pt]
        &\left.\qquad\quad+\frac{\bar{G}_{4,XX}}{\bar{G}_{2,F}^2}\,\delta A^2(\partial\,\delta A)^2\right],\\[10pt]
        &\mathcal{L}_5^{(4)}\sim \frac{\bar{G}_5}{(2\bar{G}_4)^{3/2}\,\bar{G}_{2,F}^{1/2}}\left[h(\partial h)^2(\partial\,\delta A)+h^2\,(\partial^2h)(\partial\,\delta A)+h(\partial^2 h)(\partial h)\delta A+(\partial h)^3\delta A\right]\\[5pt]
        &\qquad\quad +\frac{\bar{G}_{5,X}}{\sqrt{2\bar{G}_4}\bar{G}_{2,F}^{3/2}}\,\left[\delta A^2(\partial^2 h)(\partial\,\delta A)+h(\partial\,\delta A)^3+(\partial\,\delta A)^2(\partial h)\delta A\right]\\[5pt]
        &\qquad\quad+\frac{\bar{g}_5}{\sqrt{2\bar{G}_4}\bar{G}_{2,F}^{3/2}}\,\left[h\,\tilde{F}_{\delta A}^2(\partial\,\delta A)+h\,(\partial\,\delta A)^3+(\partial\,\delta A)^2(\partial h)\delta A\right],\\[10pt]
        &\mathcal{L}_6^{(4)}\sim\frac{\bar{G}_6}{2\bar{G}_4\,\bar{G}_{2,F}}\,\left[h\,(\partial^2h)(\partial\,\delta A)^2+(\partial h)^2(\partial\,\delta A)^2+(\partial^2 h)(\partial\,\delta A)(\partial h)\delta A\right]+\frac{\bar{G}_{6,X}}{\bar{G}_{2,F}^2}\,\tilde{F}_{\delta A}^2(\partial\,\delta A)^2.
    \end{aligned}
\end{equation}
In general, terms proportional to $\sim (\partial h)\delta A$ stem from expansions of the covariant derivative acting on the covariant Proca field $A^\mu$. The terms proportional to $\sim \bar{G}_{i,X}$ in  $\mathcal{L}_4^{(4)}$ and $\mathcal{L}_5^{(4)}$ receive not only contributions from the vector self-interactions $\sim G_{i,X}$ ($i=4,5$) in the covariant Lagrangians, but also from expanding the non-minimal coupling to gravity $\sim G_{4,5}$, resulting in terms that particularly include $\partial^2 h$. Furthermore, notice that terms proportional to $\sim \bar{G}_{2,Y}$ in $\mathcal{L}_2^{(4)}$ contribute with the same vector structure as terms $\sim G_{2,XF}$, which is why we can schematically drop the $Y$-dependence of $G_2$ in the generic model \eqref{vecgal}. Crucially, for the specific choice of this model, we see that higher order interactions beyond the vector Galileons on flat space will vanish in the decoupling limit, as they are either dressed by additional powers of $h$ that come with a heavy $M_{\rm pl}$ suppression (see e.g. $\mathcal{L}_3^{(4)}$, where $\sqrt{2\bar{G}_4}=M_{\rm pl}, \,\bar{G}_{3,X}=\frac{m^2}{\Lambda_2^2}$ and $\bar{G}_{2,F}=1$), or bring in more factors of $\frac{m^2\delta A^2}{\Lambda_2^4}$ or $\frac{F_{\delta A}^2}{\Lambda_2^4}$ that do not survive at high energies.

\bibliographystyle{JHEP}
\bibliography{biblio.bib}
\end{document}